\documentclass[runningheads]{llncs}

\usepackage[T1]{fontenc}
\usepackage{lmodern}
\usepackage{amsmath}
\usepackage{orcidlink}
\usepackage{graphicx}
\usepackage{amssymb}
\usepackage{tikz}
\usepackage{enumitem}
\usepackage{mathtools}
\usepackage{stmaryrd} 
\usepackage{mathpartir}
\usepackage{hyperref}
\usepackage[capitalise]{cleveref}
\usepackage{xspace}

\usepackage{bm}
\usepackage{marvosym}

\title{\texorpdfstring{From Rewrite Rules to Axioms \\ in the $\lambda\Pi$-Calculus Modulo Theory}{From Rewrite Rules to Axioms in the LambdaPi-Calculus Modulo Theory}}

\author{Valentin Blot\inst{1} \and
Gilles Dowek\inst{1}\orcidlink{0000-0001-6253-935X} \and
Thomas Traversié\inst{1,2} \Letter \and
Théo Winterhalter\inst{1}\orcidlink{0000-0002-9881-3696}}

\authorrunning{V. Blot et al.}

\institute{Université Paris-Saclay, Inria, ENS Paris-Saclay, CNRS, LMF, Gif-sur-Yvette, France
\email{\{valentin.blot,gilles.dowek,thomas.traversie,theo.winterhalter\}@inria.fr} \and
Université Paris-Saclay, CentraleSupélec, MICS, Gif-sur-Yvette, France}

\def\ra{\rightarrow}
\def\lra{\hookrightarrow}
\def\T{\mathcal{T}}

\def\S{\mathcal{S}}
\def\P{\mathcal{P}}
\def\E{\mathcal{E}}
\def\Type{\mbox{\tt TYPE}}
\def\Kind{\mbox{\tt KIND}}
\def\Set{{\it Set}}
\def\El{{\it El}}

\def\Prf{{\it Prf}}
\def\imp{\mathbin{\Rightarrow}}
\def\impd{\mathbin{\Rightarrow_d}}

\def\arrd{\mathbin{\rightsquigarrow_d}}
\def\o{o}
\def\blpi{\pi}
\def\fa{\forall}
\def\trad{\triangleleft}

\def\vdashr{\vdash_{\mathcal{R}}}
\def\heq#1#2#3#4{#1 ~{}_{#2}{\approx}_{#4} ~#3}
\def\pack#1#2{#1 \star #2}
\def\refl{\mathsf{refl}}
\def\trans{\mathsf{trans}}
\def\sym{\mathsf{sym}}
\def\leibp{\mathsf{leib}^\mathsf{Prf}}

\def\eqleibp{\mathsf{eqLeib}^\mathsf{Prf}}

\def\app{\mathsf{app}}
\def\fun{\mathsf{fun}}
\def\prod{\mathsf{prod}}
\def\transp{\mathsf{transp}}
\def\eqtrans{\mathsf{eqTransp}}
\def\succ{\mathsf{succ}}
\def\nat{\mathsf{nat}}
\def\lis{\mathsf{list}}
\def\nil{\mathsf{nil}}
\def\cons{\mathsf{cons}}
\def\concat{\mathsf{concat}}
\def\isrev{\mathsf{isRev}}
\def\graph{\mathsf{graph}}
\def\node{\mathsf{node}}
\def\root{\mathsf{root}}
\def\true{\mathsf{True}}

\def\eqt#1#2{\kappa(#1,#2)} 
\def\judg#1#2#3#4#5#6{#1 \vdash #2 : #3 \in \llbracket #4 \vdashr #5 : #6 \rrbracket}
\def\judgcont#1#2{\vdash #1 \in \llbracket \vdashr #2 \rrbracket}

\newcommand{\overbar}[1]{\mkern 1.5mu\overline{\mkern-1.5mu#1\mkern-1.5mu}\mkern 1.5mu}

\newcommand{\lpc}{$\lambda \Pi$-calculus\xspace}
\newcommand{\lpcm}{$\lambda \Pi$-calculus modulo theory\xspace}
\newcommand{\lpm}{$\lambda\Pi\slash{\equiv}$\xspace}

\newcommand{\ie}{\emph{i.e.}\xspace}
\newcommand{\al}{\emph{et al.}\xspace}

\begin{document}

\maketitle

\begin{abstract}
The \lpcm is an extension of simply typed $\lambda$-calculus with dependent types and user-defined rewrite rules. We show that it is possible to replace the rewrite rules of a theory of the \lpcm by equational axioms, when this theory features the notions of proposition and proof, while maintaining the same expressiveness. To do so, we introduce in the target theory a heterogeneous equality, and we build a translation that replaces each use of the conversion rule by the insertion of a transport. At the end, the theory with rewrite rules is a conservative extension of the theory with axioms.

\end{abstract}

\section{Introduction}

For Poincaré, the reasoning by which we deduce that $2 + 2 = 4$ is not a meaningful proof, but a simple verification. He concludes that the goal of exact sciences is to ``dispense with these direct verifications''~\cite{poincare}. Far from being solely a philosophical issue, this principle impacts the foundations of logical systems and in particular the choice between \emph{axioms} and \emph{rewrite rules}. For instance, in systems with axioms $x + \succ\ y = \succ\ (x + y)$ and $x + 0 = x$, we can \emph{prove} that $2 + 2 = 4$. On the other hand, in systems with rewrite rules $x + \succ\ y \lra \succ\ (x + y)$ and $x + 0 \lra x$, we just need to prove $4 = 4$ as we can \emph{compute} that $(2 + 2 = 4) \equiv (4 = 4)$. In that respect, logical systems with computation rules are convenient tools for making proofs. That is why rewrite rules have been added to systems such as \textsc{Agda}~\cite{agda} or \textsc{Coq}~\cite{coqrules} and why Dowek~\cite{dowekHDR,deductionmod} developed Deduction modulo theory, an extension of first-order logic that mixes computation and proof. Since logical systems with rewrite rules are more user-friendly, one may ask whether or not the results are the same as in axiomatic logical systems.

Rewrite rules are at the core of the \lpcm, an extension of simply typed $\lambda$-calculus with dependent types and user-definable rewrite rules~\cite{lambdapi}. The combination of $\beta$-reduction and of the rewrite rules of a signature $\Sigma$ forms the conversion $\equiv_{\beta\Sigma}$. If we know that $t : A$ with conversion $A \equiv_{\beta\Sigma} B$, then we can derive that $t : B$. In this system, a theory is a set of rewrite rules, together with a set of axioms (that are typed constants). The \lpcm is a powerful logical framework in which many theories can be expressed, such as Predicate logic, Simple type theory or the Calculus of constructions~\cite{theoryU}. It is the theory behind the \textsc{Dedukti} language~\cite{expressing,deduktiengine} and the \textsc{Lambdapi} proof assistant.

In this paper, we choose to study the replacement of rewrite rules by axioms in the \lpcm. Since it is a logical framework, the result applies to many theories. Moreover, as \textsc{Dedukti} is geared towards the interoperability between proof systems, if we want to exchange proofs between a system with rewrite rules and a system without rewrite rules \emph{via} \textsc{Dedukti}, we need to replace rewrite rules by axioms in the \lpcm. Working in this logical framework rather than in an extension of Martin-Löf type theory~\cite{martinlof} is therefore relevant on both theoretical and practical levels, but complicates the task as the \lpcm does not feature identity types or an infinite hierarchy of sorts.

One method to replace rewrite rules by axioms is to mimic the behavior of the conversion rule using transports: if we have $t : A$ and $A \equiv_{\beta\Sigma} B$ with $p$ an equality between $A$ and $B$, then we can deduce that $\transp ~p ~t : B$, but we do not directly have $t : B$. However trivial this seems, we face several challenges when trying to demonstrate it fully: the insertion of transports in terms and types is difficult due to the presence of dependent types, and the building of transports is involved as we cannot have inside the \lpcm an equality between types.

A similar problem is the elimination of equality reflection from extensional systems. Equality reflection states that $\ell = r$ implies $\ell \equiv r$, just like $\ell \lra r$ implies $\ell \equiv r$ in systems with rewrite rules. In extensional systems, typing is eased by a more powerful conversion. Hofmann~\cite{hofmann95,hofmann97} investigated categorically the problem. Oury~\cite{oury} developed a translation of proofs from an extensional version of the Calculus of Constructions to the Calculus of Inductive Constructions with equality axioms. Winterhalter, Sozeau and Tabareau~\cite{transport,WinterhalterFormalMetaType} built upon this result to reduce the number of axioms needed.

The replacement of rewrite rules by axioms paves the way for the interpretation of a theory into another inside the \lpcm. Indeed, when interpreting a theory into another, we represent each constant of the source theory by a term in the target theory, but we cannot generally do the same for rewrite rules. We can however pre-process the source theory to replace its rewrite rules by axioms, and then interpret it. The interpretation of theories allows to prove relative consistency and relative normalization theorems~\cite{realizmod}.

\paragraph{Contribution.} The main contribution of this paper is the translation of a theory with rewrite rules to a theory with equational axioms. To do so, we restrict the theories considered to theories with an encoding of the notions of proposition and proof inside the \lpcm. So as to compare objects that possibly do not have the same type, we define a heterogeneous equality---following the one defined by McBride~\cite{mcbride}. The restriction considered allows us to build an equality between particular types---called small types. We define a type system with typed conversion for the \lpcm, so that the proofs are done by induction on the derivation trees more easily.

\paragraph{Outline of the paper.} In \cref{part_lpcm}, we present the \lpcm, we detail a prelude encoding of the notions of proposition and proof in it, and we identify the assumptions made on the considered theories. The heterogeneous equality and the equality between small types are presented in \cref{part_eq}. The replacement of rewrite rules by axioms and the translation of terms, judgments and theories are presented in \cref{part_trans}.

\section{Theories in the \texorpdfstring{$\lambda\Pi$}{lambdaPi}-Calculus Modulo Theory}
\label{part_lpcm}

In this section, we give a more detailed overview of the \lpcm~\cite{lambdapi} and its type system. In particular, we present an encoding of the notions of proposition and proof in the \lpcm~\cite{theoryU}. We characterize small types---a subclass of types for which we can define an equality.

\subsection{The \texorpdfstring{$\lambda\Pi$}{lambdaPi}-Calculus Modulo Theory}

The \lpc, also known as the Edinburgh Logical Framework~\cite{LF}, is an extension of simply typed $\lambda$-calculus with dependent types. The \lpcm (\lpm)~\cite{lambdapi} is an extension of the \lpc, in which user-definable rewrite rules have been added~\cite{rewriteSystem}. Its syntax is given by:
\begin{align*}
&Sorts &&s \Coloneqq \Type ~|~ \Kind \\
&Terms &&t,u, A, B \Coloneqq c ~|~ x ~|~ s ~|~ \Pi x : A. ~B ~|~ \lambda x : A. ~t ~|~ t ~u \\
&Contexts &&\Gamma \Coloneqq \langle \rangle ~|~ \Gamma, x : C \\
&Signatures &&\Sigma \Coloneqq \langle \rangle ~|~ \Sigma, c : D ~|~ \Sigma, \ell \lra r
\end{align*}
where $c$ is a constant and $x$ is a variable (ranging over disjoint sets), $C$ and $r$ are terms, $D$ is a closed term (\ie a term with no free variables) and $\ell$ is a term such that $\ell = c ~t_1 \ldots t_k$ with $c$ a constant. $\Type$ and $\Kind$ are two sorts: terms of type $\Type$ are called types, and terms of type $\Kind$ are called kinds. $\Pi x : A. ~B$ is a dependent product, $\lambda x : A. ~t$ is an abstraction and $t ~u$ is an application. $\Pi x : A. ~B$ is simply written $A \ra B$ if $x$ does not appear in $B$. Signatures and contexts are finite sequences, and are written $\langle \rangle$ when empty. Signatures contain both typed constants and rewrite rules (written $\ell \lra r$). \lpm is a logical framework, in which $\Sigma$ is fixed by the user depending on the logic they are working in.

The relation $\lra_{\beta\Sigma}$ is generated by $\beta$-reduction and by the rules of $\Sigma$. More explicitly, $\lra_{\beta\Sigma}$ is the smallest relation, closed by context, such that if $t$ rewrites to $u$ for some rule in $\Sigma$ or by $\beta$-reduction then $t \lra_{\beta\Sigma} u$. Conversion $\equiv_{\beta\Sigma}$ is the reflexive, symmetric, and transitive closure of $\lra_{\beta\Sigma}$.

\subsection{The Type System of the \texorpdfstring{$\lambda\Pi$}{lambdaPi}-Calculus Modulo Theory}

We introduce in \cref{typ_lpcm,conv_lpcm} typing rules for \lpm. \cref{typ_lpcm} presents the usual typing rules while \cref{conv_lpcm} focuses on the conversion rules. We write $\vdash \Gamma$ when the context $\Gamma$ is well formed and $\Gamma \vdash t : A$ when $t$ is of type $A$ in the context $\Gamma$. $\langle \rangle \vdash t : A$ is simply written $\vdash t : A$. The notation $(\vdash \Gamma_1) \equiv (\vdash \Gamma_2)$ means that $\Gamma_1$ and $\Gamma_2$ are both well formed, have the same length and have the same variables with convertible types.
We write $(\Gamma_1 \vdash t_1 : A_1) \equiv (\Gamma_2 \vdash t_2 : A_2)$ when $t_1$ and $t_2$ are convertible with $\Gamma_1 \vdash t_1 : A_1$ and $\Gamma_2 \vdash t_2 : A_2$. In particular, convertible terms $t_1 \equiv t_2$ are authorized to have different types---provided that both types are convertible---and to be typed in different contexts---provided that both contexts are convertible. In \textsc{ConvRule}, $\vec{x}$ is a vector representing the free variables of $\ell$. The standard weakening rule and substitution lemma can be derived from this type system.

\begin{figure}
\begin{mathpar}
\inferrule*[right={[Empty]}]{ }{\vdash \langle \rangle}

\inferrule*[right={[Decl] $x \notin \Gamma$}]{\vdash \Gamma \\ \Gamma \vdash A : s}{\vdash \Gamma, x : A}

\inferrule*[right={[Sort]}]{\vdash \Gamma}{\Gamma \vdash \Type : \Kind}

\inferrule*[right={[Const] $c : A \in \Sigma$}]{\vdash \Gamma \\ \vdash A : s}{\Gamma \vdash c : A}

\inferrule*[right={[Var] $x : A \in \Gamma$}]{\vdash \Gamma}{\Gamma \vdash x : A}

\inferrule*[right={[Prod]}]{\Gamma \vdash A : \Type \\ \Gamma, x : A \vdash B : s}{\Gamma \vdash \Pi x : A. ~B : s}

\inferrule*[right={[Abs]}]{\Gamma \vdash A : \Type \\ \Gamma, x : A \vdash B : s \\ \Gamma, x : A \vdash t : B}{\Gamma \vdash \lambda x : A. ~t : \Pi x : A. ~B}

\inferrule*[right={[App]}]{\Gamma \vdash t : \Pi x : A. ~B \\ \Gamma \vdash u : A}{\Gamma \vdash t\ u : B[x \mapsto u]}

\inferrule*[right={[Conv]}]{\Gamma \vdash t : A \\ (\Gamma \vdash A : s) \equiv (\Gamma \vdash B : s)}{\Gamma \vdash t : B}
\end{mathpar}

\caption{Typing rules of the \lpcm}
\label{typ_lpcm}
\end{figure}

\begin{figure}
\begin{mathpar}
\inferrule*[right={[ConvRefl]}]{\Gamma \vdash u : A}{(\Gamma \vdash u : A) \equiv (\Gamma \vdash u : A)}

\inferrule*[right={[ConvSym]}]{(\Gamma \vdash u : A) \equiv (\Gamma \vdash v : B)}{(\Gamma \vdash v : B) \equiv (\Gamma \vdash u : A)}

\inferrule*[right={[ConvTrans]}]{(\Gamma \vdash u : A) \equiv (\Gamma \vdash v : B) \\ (\Gamma \vdash v : B) \equiv (\Gamma \vdash w : C)}{(\Gamma \vdash u : A) \equiv (\Gamma \vdash w : C)}

\inferrule*[right={[ConvDecl] $x \notin \Gamma_1, \Gamma_2$}]{(\vdash \Gamma_1) \equiv (\vdash \Gamma_2) \\ (\Gamma_1 \vdash A_1 : s) \equiv (\Gamma_2 \vdash A_2 : s)}{(\vdash \Gamma_1, x : A_1) \equiv (\vdash \Gamma_2, x : A_2)}

\inferrule*[right={[ConvConst] $c : A \in \Sigma$}]{(\vdash \Gamma_1) \equiv (\vdash \Gamma_2) \\ \vdash A : s}{(\Gamma_1 \vdash c : A) \equiv (\Gamma_2 \vdash c : A)}

\inferrule*[right={[ConvVar] $x : A_1 \in \Gamma_1, x : A_2 \in \Gamma_2$}]{(\vdash \Gamma_1) \equiv (\vdash \Gamma_2)}{(\Gamma_1 \vdash x : A_1) \equiv (\Gamma_2 \vdash x : A_2)}

\inferrule*[right={[ConvProd]}]{(\Gamma_1 \vdash A_1 : \Type) \equiv (\Gamma_2 \vdash A_2 : \Type) \\ (\Gamma_1, x : A_1 \vdash B_1 : s) \equiv (\Gamma_2, x : A_2 \vdash B_2 : s)}{(\Gamma_1 \vdash \Pi x : A_1. ~B_1 : s) \equiv (\Gamma_2 \vdash \Pi x : A_2. ~B_2 : s)}

\inferrule*[right={[ConvAbs]}]{(\Gamma_1 \vdash A_1 : \Type) \equiv (\Gamma_2 \vdash A_2 : \Type) \\ (\Gamma_1, x : A_1 \vdash B_1 : s) \equiv (\Gamma_2, x : A_2 \vdash B_2 : s) \\ (\Gamma_1, x : A_1 \vdash t_1 : B_1) \equiv (\Gamma_2, x : A_2 \vdash t_2 : B_2)}{(\Gamma_1 \vdash \lambda x : A_1. ~t_1 : \Pi x : A_1. ~B_1) \equiv (\Gamma_2 \vdash \lambda x : A_2. ~t_2 : \Pi x : A_2. ~B_2)}

\inferrule*[right={[ConvApp]}]{(\Gamma_1 \vdash t_1 : \Pi x : A_1. ~B_1) \equiv (\Gamma_2 \vdash t_2 : \Pi x : A_2. ~B_2) \\ (\Gamma_1 \vdash u_1 : A_1) \equiv (\Gamma_2 \vdash u_2 : A_2)}{(\Gamma_1 \vdash t_1 ~u_1 : B_1[x \mapsto u_1]) \equiv (\Gamma_2 \vdash t_2 ~u_2 : B_2[x \mapsto u_2])}

\inferrule*[right={[ConvBeta]}]{\Gamma \vdash A : \Type \\ \Gamma, x : A \vdash t : B \\ \Gamma, x : A \vdash B : s \\ \Gamma \vdash u : A}{(\Gamma \vdash (\lambda x : A. ~t) ~u : B[x \mapsto u]) \equiv (\Gamma \vdash t[x \mapsto u] : B[x \mapsto u])}

\inferrule*[right={[ConvRule] $\ell \lra r \in \Sigma$}]{\bm{x} : \bm{B} \vdash \ell : A \\ \bm{x} : \bm{B} \vdash r : A \\ \Gamma \vdash \bm{t} : \bm{B}}{(\Gamma \vdash \ell[\bm{x} \mapsto \bm{t}] : A[\bm{x} \mapsto \bm{t}]) \equiv (\Gamma \vdash r[\bm{x} \mapsto \bm{t}] : A[\bm{x} \mapsto \bm{t}])}

\inferrule*[right={[ConvConv]}]{\Gamma \vdash u  : A \\ (\Gamma \vdash A : s) \equiv (\Gamma \vdash B : s)}{(\Gamma \vdash u  : A) \equiv (\Gamma \vdash u : B)}
\end{mathpar}

\caption{Convertibility rules of the \lpcm}
\label{conv_lpcm}
\end{figure}

\begin{lemma}[Substitution]
\label{lemma_subst}
\begin{itemize}
\item If we have $\vdash \Gamma, x : A, \Delta$ and $\Gamma \vdash u : A$, then $\vdash \Gamma, \Delta[x \mapsto u]$.
\item If we have $\Gamma, x : A, \Delta \vdash t : B$ and $\Gamma \vdash u : A$, then $\Gamma, \Delta[x \mapsto u] \vdash t[x \mapsto u] : B[x \mapsto u]$.
\item If we have $(\vdash \Gamma_1, x : A_1, \Delta_1) \equiv (\vdash \Gamma_2, x : A_2, \Delta_2)$ and $\Gamma_1 \vdash u : A_1$, then $(\vdash \Gamma_1, \Delta_1[x \mapsto u]) \equiv (\vdash \Gamma_2, \Delta_2[x \mapsto u])$.
\item If we have $(\Gamma_1, x : A_1, \Delta_1 \vdash t_1 : B_1) \equiv (\Gamma_2, x : A_2, \Delta_2 \vdash t_2 : B_2)$ and $\Gamma_1 \vdash u : A_1$, then $(\Gamma_1, \Delta_1[x \mapsto u] \vdash t_1[x \mapsto u] : B_1[x \mapsto u]) \equiv (\Gamma_2, \Delta_2[x \mapsto u] \vdash t_2[x \mapsto u] : B_2[x \mapsto u])$.
\end{itemize}
\end{lemma}

\begin{proof}
We proceed by induction on the typing derivation.
\end{proof}
We chose to present a type system with \emph{typed} conversion (written $\equiv$)---so as to easily do proofs on the derivations---while the usual type system for \lpm features \emph{untyped} conversion (written $\equiv_{\beta\Sigma}$). The equivalence between type systems with typed conversion and type systems with untyped conversion has been a longstanding question: Geuvers and Werner~\cite{geuvers_werner} investigated the case of Pure Type Systems with $\beta\eta$-convertibility, Adams~\cite{adams} proved the equivalence in the case of functional Pure Type Systems, and Siles~\cite{silesphd,siles} later proved the equivalence in the general case of the Pure Type Systems. The case of \lpm, in which we have $\beta$-convertibility but also user-defined rewrite rules, remains to be investigated.

We write $| \Sigma |$ for the set of constants of $\Sigma$, and $\Lambda(\Sigma)$ for the set of terms $t$ whose constants belong to $| \Sigma |$. We say that $\T = \Sigma$ is a theory when for each rule $\ell \lra r \in \Sigma$ we have $\ell$ and $r$ in $\Lambda(\Sigma)$, when $\lra_{\beta\Sigma}$ is confluent on $\Lambda(\Sigma)$, and when every rule of $\Sigma$ preserves typing in $\Sigma$ (that is when for all context $\Gamma$ and for all term $A \in \Lambda(\Sigma)$, if $\Gamma \vdash \ell : A$ then $\Gamma \vdash r : A$).

\begin{example}[Natural numbers and lists]
We can define in \lpm a partial theory of natural numbers and indexed lists of natural numbers. $\nat$ represents the type of natural numbers and $\lis$ represents the dependent type of indexed lists of natural numbers. $\cons$ adds a new element to a list, $\concat$ concatenates two lists, and $\isrev$ checks if the first given list is the reverse of the second.
\begin{mathpar}
  \nat : \Type

  0 : \nat

  \succ : \nat \ra \nat

  + : \nat \ra \nat \ra \nat

  x + 0 \lra x

  x + \succ\ y \lra \succ\ (x + y)

  \lis : \nat \ra \Type

  \nil : \lis\ 0

  \cons : \Pi x : \nat.\ \lis\ x \ra \nat \ra \lis\ (\succ\ x)

  \isrev : \Pi x : \nat.\ \lis\ x \ra \lis\ x \ra \Type

  \concat : \Pi x,y : \nat.\ \lis\ x \ra \lis\ y \ra \lis\ (x + y)
\end{mathpar}
In the context $\ell : \lis\ (\succ\ 0)$, we have $\concat\ (\succ\ 0)\ 0\ \ell\ \nil$ of type $\lis\ (\succ\ 0 + 0)$. If we want to compare $\ell$ and this new list with $\isrev$, we cannot directly do it because they do not have the same type. However, we can use the conversion rule with $\lis\ (\succ\ 0 + 0) \equiv_{\beta\Sigma} \lis\ (\succ\ 0)$. This conversion derives from the rewrite rule $x + 0 \lra x$ instantiated with $x \coloneqq \succ\ 0$.
\end{example}

\subsection{A Prelude Encoding for the \texorpdfstring{$\lambda\Pi$}{lambdaPi}-Calculus Modulo Theory}

It is possible to introduce in \lpm the notions of proposition and proof~\cite{theoryU}. In particular, this encoding---called prelude encoding---gives the possibility to quantify on certain propositions through codes, which is not possible inside the standard \lpm. This encoding is defined by following signature.

\begin{definition}
The signature $\Sigma_{pre}$ contains the following constants and rewrite rules:
\begin{flalign*}
&\Set : \Type  & &\o : \Set \\
&\El : \Set \ra \Type & &\Prf : \El ~\o \ra \Type  \\
&\arrd : \Pi x : \Set. ~(\El ~x \ra \Set) \ra \Set & &\impd : \Pi x : \El ~\o. ~(\Prf ~x \ra \El ~\o) \ra \El ~\o \\
&\blpi : \Pi x : \El ~\o. ~(\Prf ~x \ra \Set) \ra \Set & &\fa : \Pi x : \Set. ~(\El ~x \ra \El ~\o) \ra \El ~\o \\
&\El ~(x \arrd y) \lra \Pi z : \El ~x. ~\El ~(y ~z) & &\Prf ~(x \impd y) \lra \Pi z : \Prf ~x. ~\Prf ~(y ~z) \\
&\El ~(\blpi ~x ~y) \lra \Pi z : \Prf ~x. ~\El ~(y ~z) & &\Prf ~(\fa ~x ~y) \lra \Pi z : \El ~x. ~\Prf ~(y ~z)
\end{flalign*}

\end{definition}
We declare the constant $\Set$, which represents the universe of types, along with the injection $\El$ that maps terms of type $\Set$ into $\Type$. $\o$ is a term of type $\Set$ such that $\El ~\o$ defines the universe of propositions. The injection $\Prf$ maps propositions into $\Type$. $\arrd$ (respectively $\impd$) is written infix and is used to represent dependent function types between terms of type $\Set$ (respectively $\El ~\o$). The symbol $\blpi$ (respectively $\fa$) is used to represent dependent function types between elements of type $\El ~\o$ and $\Set$ (respectively $\Set$ and $\El ~\o$).

The main advantage of this encoding is that it allows us to quantify on propositions. Indeed, in \lpm, we cannot quantify on $\Type$. Instead, we can quantify on objects of type $\El ~\o$, and then inject them into $\Type$ using $\Prf$.

\subsection{Small Types and Small Derivations}

As we work in \lpm rather than in an extension of Martin-Löf type theory, we do not have a pre-defined equality. Moreover, we cannot define an equality between types since such object would have type $\Type \ra \Type \ra \Type$, which is not allowed in \lpm.

If we want to compare types $\Prf ~a$ and $\Prf ~b$, we cannot do it directly, but we can compare $a$ and $b$ (that are of type $\El ~\o$). We can proceed similarly to compare types $\El ~a$ and $\El ~b$ (with $a$ and $b$ of type $\Set$). In that respect, we want types to be into a special form---called small type---that takes advantages of the prelude encoding, so as to compare them if necessary. To put types of the prelude encoding into this special form, we use the reverse of the rewrite rules of $\Sigma_{pre}$ to represent dependent types with the symbols $\arrd$, $\impd$, $\blpi$ and $\fa$ whenever it is possible. This is achieved by the partial function $\nu$, defined by:
\begin{mathpar}
\nu(\Set) = \Set

\nu(\Prf ~a) = \Prf ~a

\nu(\El ~a) = \El ~a

\begin{array}{lcll}
    \nu(\Pi x : A. ~B) &=&
    \Prf ~(a \impd (\lambda x : \Prf ~a. ~b)) & \text{ if $\nu(A) = \Prf ~a$ and $\nu(B) = \Prf ~b$} \\
    && \El ~(a \arrd (\lambda x : \El ~a. ~b)) & \text{ if $\nu(A) = \El ~a$ and $\nu(B) = \El ~b$} \\
    && \Prf ~(\fa ~a ~(\lambda x : \El ~a. ~b)) & \text{ if $\nu(A) = \El ~a$ and $\nu(B) = \Prf ~b$} \\
    && \El ~(\blpi ~a ~(\lambda x : \Prf ~a. ~b)) & \text{ if $\nu(A) = \Prf ~a$ and $\nu(B) = \El ~b$} \\
    && \Pi x : \nu(A). ~\nu(B) & \text{ otherwise}
  \end{array}
\end{mathpar}
Therefore, when $\nu(A)$ is defined, we have $A \equiv_{\beta\Sigma_{pre}} \nu(A)$. Note that $\nu$ is partial because we do not handle the case where a type is a $\beta$-reducible expression, as in practice we will not have types under $\lambda$-abstraction form.

To continue to characterize a particular form of types, we define the three following grammars:
\begin{mathpar}
\S \Coloneqq \Set \mid \S \ra \S

\P \Coloneqq \Prf ~a \mid \P \ra \S \mid \Pi z : \S.\ \P

\E \Coloneqq \El ~b \mid \E \ra \S \mid \Pi z : \S.\ \E
\end{mathpar}
with $a : \El ~\o$ and $b : \Set$. The notation $A \in \S$ means that $A$ is generated by the grammar $\S$. The grammar $\S$ generates types that only contain $\Set$. Therefore, if $\nu(A) \in \S$ then $\nu(A) = A$. The grammars $\P$ and $\E$ generate types that contain a central symbol $\Prf$ or $\El$.

\begin{definition}[Small type, Small context]
A type $A$ is small when $\nu(A)$ is defined and $\nu(A) \in \S \cup \P \cup \E$. In that case, $\nu(A)$ is called the small form of $A$. A context $\Gamma$ is small when for every $x : A \in \Gamma$ we have that $A$ is a small type.
\end{definition}
\begin{example}
$\Prf ~a \ra \Prf ~b$, with $a, b : \El ~\o$, is a small type since its small form
$\Prf ~(a \impd (\lambda z. ~b))$
is generated by the grammar $\P$. The type $\Pi x : \Prf ~b. ~\El ~c$, with $c : \Set$ depending on $x$, is a small type since its small form $\El ~(\blpi ~b ~(\lambda x : \Prf ~b. ~c))$ is generated by the grammar $\E$. The type $\Prf ~a \ra \Set \ra \Prf ~b$ is not small, since $\nu(\Prf ~a \ra \Set \ra \Prf ~b) = \Prf ~a \ra \Set \ra \Prf ~b \notin \S \cup \P \cup \E$.
\end{example}
We would ideally like all the types to be small, so that we can compare them if necessary. Therefore, if $\Gamma \vdash t : A$, we want $A$ to be a small type, or $t$ to be a small type and $A = \Type$. However, small types are built using the constants of $\Sigma_{pre}$. In particular, the type of the constants $\o$, $\arrd$, $\impd$ and $\fa$ are small, but the types of $\blpi$, $\Prf$ and $\El$ are not. Note that the type of an application of $\blpi$, $\Prf$ or $\El$ is small. We thus come up with the following notion.

\begin{definition}[Small judgment]
$\vdash \Gamma$ is a small judgment when $\Gamma$ is a small context.
$\Gamma \vdash t : A$ is a small judgment when $\Gamma$ is a small context and when
\begin{itemize}
\item $t : A \in \Sigma_{pre}$,
\item or $t$ is the type of a constant of $\Sigma_{pre}$,
\item or $A$ is a small type,
\item or $t$ is a small type.
\end{itemize}
$(\Gamma_1 \vdash t_1 : A_1) \equiv (\Gamma_2 \vdash t_2 : A_2)$ is a small judgment when $\Gamma_1 \vdash t_1 : A_1$ and $\Gamma_2 \vdash t_2 : A_2$ are small.
\end{definition}

\begin{definition}[Small derivation]
A small derivation is a derivation in which all the judgments are small.
\end{definition}

\subsection{Theories with Prelude Encoding}

We define the theories we will consider in the rest of the paper: theories that features the prelude encoding inside \lpm.

\begin{definition}[Theory with prelude encoding]
We say that a theory $\T = \Sigma$ in the \lpm is a theory with prelude encoding when:
\begin{itemize}
\item there exists $\Sigma_\T$ such that $\Sigma = \Sigma_{pre} \cup \Sigma_\T$ and $\Sigma_{pre} \cap \Sigma_\T = \emptyset$,
\item for every $c : A \in \Sigma_\T$, $A$ is small and admits a small derivation $\vdash A : \Type$,
\item for every $\ell \lra r \in \Sigma_\T$, we have small derivations $\bm{x} : \bm{B} \vdash \ell : A$ and $\bm{x} : \bm{B} \vdash r : A$ with $A$ a small type, where $\bm{x}$ represents the free variables of $\ell$.
\end{itemize}
\end{definition}
A theory with prelude encoding is a theory with the constants and rewrite rules $\Sigma_{pre}$, and additional user-defined constants and rewrite rules. To ensure that $\Sigma_\T$ is encoded \emph{inside} the prelude encoding, we can only define new constants whose types are small. We do not allow the use of rewrite rules $\ell \lra r$ when $\ell$ has $\Type$ in its type. In particular, we cannot define new rewrite rules on $\Prf$ or $\El$ and change the behavior of these constants. It follows that the three grammars $\S$, $\P$ and $\E$ generate disjoint types.

In the following examples, we present three theories with prelude encoding in \lpm. The examples of predicate logic and set theory illustrate that the restrictions considered are generally respected, even for expressive theories.

\begin{example}[Predicate logic]
Predicate logic can be encoded in a theory with prelude encoding. We declare constants for tautology and contradiction $\top, \bot : \El ~\o$, for negation $\neg : \El ~\o \ra \El ~\o$, for conjunction and disjunction $\wedge, \vee : \El ~\o \ra \El ~\o \ra \El ~\o$, and for existential quantification  $\exists : \Pi z : \Set. ~(\El ~z \ra \El ~\o) \ra \El ~\o$. The semantics of tautology is defined by the rewrite rule $\top \lra \fa ~\o ~(\lambda x : \El ~\o. ~x \imp x)$, which is equivalent to the more common form $\Prf ~\top \lra \Pi z : \El ~\o. ~\Prf ~z \ra \Prf ~z$. The rewrite rule $\Prf ~(A \wedge B) \lra \Pi P : \El ~\o. ~(\Prf ~A \ra \Prf ~B \ra \Prf ~P) \ra \Prf ~P$ can be encoded by $A \wedge B \lra \fa ~\o ~(\lambda P. ~(A \ra B \ra P) \ra P)$. The rule $\Prf ~(\neg A) \lra \Prf ~A \ra \Prf ~\bot$ is forbidden, but $\neg A \lra A \imp \bot$ is allowed. We proceed similarly the other rewrite rules.
\end{example}

\begin{example}[Natural numbers and lists]
\label{ex_vec}
We can define our small theory of natural numbers and lists in the prelude encoding, by replacing $\Type$ by $\Set$ (in the universe of types) or $\El ~\o$ (in the universe of propositions), and by adding $\El$ and $\Prf$ at the necessary positions.
\begin{mathpar}
\nat : \Set

0 : \El\ \nat

\succ : \El\ \nat \ra \El\ \nat

+ : \El\ \nat \ra \El\ \nat \ra \El\ \nat

\lis : \El\ \nat \ra \Set

x + 0 \lra x

x + \succ\ y \lra \succ\ (x + y)

\nil : \El\ (\lis\ 0)

\cons : \Pi x : \El\ \nat.\ \El\ \lis\ x \ra \El\ \nat \ra \El\ (\lis\ (\succ\ x))

\isrev : \Pi x : \El\ \nat.\ \El\ (\lis\ x) \ra \El\ (\lis\ x) \ra \El\ \o

\concat : \Pi x,y : \El\ \nat.\ \El\ (\lis\ x) \ra \El\ (\lis\ y) \ra \El\ (\lis\ (x + y))
\end{mathpar}
\end{example}

\begin{example}[Set theory]
The implementation in \textsc{Dedukti} of set theory~\cite{deduktiz} is a theory with prelude encoding. In this implementation, sets are represented by a more primitive notion of pointed graphs: we have $\graph$ and $\node$ of type $\Set$. The predicate $\eta : \El\ \graph \ra \El\ \node \ra \El\ \node \ra \El\ \o$ is such that $\eta\ a\ x\ y$ is the proposition asserting that there is an edge in $a$ from $y$ to $x$. The operator $\root : \El\ \graph \ra \El\ \node$ returns the root of a, which is a node.
\end{example}
In practice, the derivations of small judgments are small derivations. As we consider theories with prelude encoding, the only way of introducing a judgment that is not small is through $\lambda$-abstractions. For instance in \cref{ex_vec} the judgment $\vdash \El ~(\lis ~((\lambda x : \El ~\nat. ~\lambda y : \Set. ~x) ~0 ~\nat)) : \Type$ is small, but in its derivation we have $\vdash \lambda x : \El ~\nat. ~\lambda y : \Set. ~x : \El ~\nat \ra \Set \ra \El ~\nat$ which is not a small judgment. However, $\vdash \El ~(\lis ~0) : \Type$ admits a small derivation. If the derivation is not small, we can in practice apply $\beta$-reduction on the fragments of the derivation that are not small to obtain a small derivation.

\section{Equalities}
\label{part_eq}

Since we want to replace rewrite rules $\ell \lra r$ by equational axioms $\ell = r$, we need to define an equality in the target theory. In this section, we present a heterogeneous equality and a method to compare small types. The heterogeneous equality is necessary to compare objects that do not have the same type. Although we cannot define an equality between types in \lpm, it is possible to develop an equality between small types, taking advantage of their structure.

\subsection{Heterogeneous Equality}

In our development, we need to have an equality between two translations of the same term. However, the two translations do not necessarily have the same type, as we may have introduced transports over the course of the translation. To that end, we define a heterogeneous equality inspired by the one of McBride~\cite{mcbride}. Our heterogeneous equality is defined by the constant schemas $\mathsf{heq}_{A,B} : A \ra B \ra \El ~\o$ where $A$ and $B$ are of type $\Type$.
We write
$\heq{u}{A}{v}{B}$ for $\Prf ~(\mathsf{heq}_{A,B} ~u ~v)$. Heterogeneous equality is reflexive, symmetric, and transitive.
\[
  \begin{array}{lcl}
	&&\refl_A :
	\Pi u : A. ~\heq{u}{A}{u}{A} \\
	&&\sym_{A,B} : \Pi u : A. ~\Pi v : B. ~\heq{u}{A}{v}{B} \ra \heq{v}{B}{u}{A} \\
	&&\trans_{A,B,C} : \Pi u : A. ~\Pi v : B. ~\Pi w : C.\ \heq{u}{A}{v}{B} \ra \heq{v}{B}{w}{C} \ra \heq{u}{A}{w}{C}
  \end{array}
\]
When two objects have the same type, heterogeneous equality acts as Leibniz equality. In particular, we can replace $u$ by $v$ in the universes of propositions and types. The result of a Leibniz substitution on $t$ remains equal to $t$.
\[
  \begin{array}{lcl}
	\leibp_A &:& \Pi u,v : A. ~\Pi p : \heq{u}{A}{v}{A}. ~\Pi P : A \ra \El ~\o. ~\Prf ~(P ~u) \ra \Prf ~(P ~v) \\
	\eqleibp_A &:& \Pi u,v : A. ~\Pi p : \heq{u}{A}{v}{A}. ~\Pi P : A \ra \El ~\o. ~\Pi t : \Prf ~(P ~u). \\
	&&\heq{\leibp_A ~u ~v ~p ~P ~t}{\Prf ~(P ~v)}{t}{\Prf ~(P ~u)}
  \end{array}
\]
The same axiom schemas exist for the universe of types, with superscript $\mathsf{El}$ instead of $\mathsf{Prf}$, $\El$ instead of $\Prf$, and $\Set$ instead of $\El ~\o$.

Finally, we add axioms for the congruence of each constructor of \lpm.

\paragraph{Application constructor.} For the application, we take:
\[
  \begin{array}{lcl}
	\app_{A_1,A_2,B_1,B_2} &:&
	\Pi t_1 : (\Pi x : A_1. ~B_1). ~\Pi t_2 : (\Pi x : A_2. ~B_2). \\
	&&\Pi u_1 : A_1. ~\Pi u_2 : A_2. ~\heq{t_1}{}{t_2}{} \ra \heq{u_1}{}{u_2}{} \\
	&&\ra \heq{t_1 ~u_1}{B_1[x \mapsto u_1]}{t_2 ~u_2}{B_2[x \mapsto u_2]}
  \end{array}
\]
For the $\lambda$-abstraction and $\Pi$-type constructors, we cannot directly build equality axioms. Indeed, if we want to define an equality between functional terms $t_1$ of type $\Pi x : A_1. ~B_1$ and $t_2$ of type $\Pi x : A_2. ~B_2$, we need to ensure that types $A_1$ and $A_2$ are equal. Therefore, we would like to have
\[
  \begin{array}{lcl}
	\fun_{A_1,A_2,B_1,B_2} &:&
	\Pi t_1 : (\Pi x : A_1. ~B_1). ~\Pi t_2 : (\Pi y : A_2. ~B_2). ~\heq{A_1}{}{A_2}{} \\
	&& \ra (\Pi x : A_1. ~\Pi y : A_2. ~\heq{x}{}{y}{} \ra \heq{t_1 ~x}{}{t_2 ~y}{}) \\
	&& \ra \heq{t_1}{}{t_2}{}
  \end{array}
\]
but we cannot take such an axiom, since the heterogeneous equality is not defined to compare objects that have type $\Type$, and $\heq{A_1}{}{A_2}{}$ is therefore ill typed. This shortcoming is addressed by developing an equality between small types.

\subsection{Equality between Small Types}

We cannot build an equality between types, since such an equality would have type $\Type \ra \Type \ra \Type$, which is impossible in \lpm. An option would be to take axiom schemas $\heq{A}{}{B}{}$ for every equality between types $A$ and $B$. Such an equality would be too far from standard and would require additional axioms to build transports. An alternative is to define an equality between small types. By construction, if $\nu(A) \in \P$, then $\nu(A)$ is generated from $\Prf ~a$ for some $a : \El ~\o$, and if $\nu(A) \in \E$, then $\nu(A)$ is generated from $\El ~a$ for some $a : \Set$. If the small form of $A$ contains $\Prf ~a$ and the small form of $B$ contains $\Prf ~b$, then we want an equality between $a$ and $b$. We define the partial function $\kappa$ on small forms by
\begin{mathpar}
\eqt{\Prf ~a_1}{\Prf ~a_2} = \heq{a_1}{}{a_2}{}

\eqt{\El ~a_1}{\El ~a_2} = \heq{a_1}{}{a_2}{}

\eqt{S}{S} = \true \text{ if $S \in \S$}

\eqt{T_1 \ra S}{T_2 \ra S} = \eqt{T_1}{T_2} \text{ if $S \in \S$}

\eqt{\Pi z : S. ~T_1}{\Pi z : S. ~T_2} = \Pi z : S. ~\eqt{T_1}{T_2} \text{ if $S \in \S$}
\end{mathpar}
where $\true \coloneqq \Pi P : \El ~\o. ~\Prf ~P \ra \Prf ~P$, so we can always give a witness of $\eqt{S}{S}$ if $S \in \S$. By convention, we simply write $\eqt{A}{B}$ for the result of $\eqt{\nu(A)}{\nu(B)}$.
\begin{example}
$\eqt{\Pi x : \Set. ~\Prf ~P \ra \Prf ~Q}{\Pi x : \Set. ~\Prf ~R} = \Pi x : \Set. ~\heq{(P \impd \lambda z : P. ~Q)}{}{R}{}$ since $\nu(\Pi x : \Set. ~\Prf ~P \ra \Prf ~Q) = \Pi x : \Set. ~\Prf ~(P \impd (\lambda z : P. ~Q))$.
\end{example}
We can now go back to the definition of equality axioms for the constructors of \lpm.

\paragraph{Function constructor.} If $A_1$ and $A_2$ are small types,
we can take $\eqt{A_1}{A_2}$. We do not compare objects of type $\Type$ anymore, but objects that have either type $\El ~\o$ or type $\Set$. The axiom schema for the function constructor is thus:
\[
  \begin{array}{lcl}
	\fun_{A_1,A_2,B_1,B_2} &:& \Pi t_1 : (\Pi x : A_1. ~B_1). ~\Pi t_2 : (\Pi y : A_2. ~B_2). ~\eqt{A_1}{A_2} \\
	&& \ra (\Pi x : A_1. ~\Pi y : A_2. ~\heq{x}{}{y}{} \ra \heq{t_1 ~x}{}{t_2 ~y}{}) \\
	&&\ra \heq{t_1}{}{t_2}{}
  \end{array}
\]
This axiom schema is a generalization of the \emph{functional extensionality} principle with distinct domains $A_1$ and $A_2$ in the case of heterogeneous equality. Functional extensionality states that two pointwise-equal functions are equal. If the domains $A_1$ and $A_2$ are generated by $\S$, then they are syntactically equal and we can derive a simpler axiom schema:
\[
  \begin{array}{lcl}
	\fun_{A,B_1,B_2} &:&
	\Pi t_1 : (\Pi x : A. ~B_1). ~\Pi t_2 : (\Pi x : A. ~B_2). ~(\Pi x : A. ~\heq{t_1 ~x}{}{t_2 ~x}{}) \\
	&& \ra \heq{t_1}{}{t_2}{}
  \end{array}
\]
\paragraph{$\Pi$-type constructor.} The congruence axiom for dependent types aims at building $\eqt{\Pi x : A_1. ~B_1}{\Pi x : A_2 ~B_2}$. There are different cases depending on the grammars generating $\nu(A_1)$, $\nu(A_2)$, $\nu(B_1)$ and $\nu(B_2)$. If $\nu(A_1)$, $\nu(A_2)$, $\nu(B_1)$, $\nu(B_2) \in \S$, then $\Pi x : A_1. ~B_1$ and $\Pi x : A_2. ~B_2$ are syntactically equal and we can build an object of type $\true$. If $\nu(A_1)$, $\nu(A_2) \in \S$ and $\nu(B_1)$, $\nu(B_2) \in \P \cup \E$, then $A_1 = A_2$ and $\eqt{\Pi x : A_1. ~B_1}{\Pi x : A_2 ~B_2} = \Pi x : A_1. ~\eqt{B_1}{B_2}$. If $\nu(A_1)$, $\nu(A_2) \in \P \cup \E$ and $\nu(B_1)$, $\nu(B_2) \in \S$, then $B_1 = B_2$ and $\eqt{\Pi x : A_1. ~B_1}{\Pi x : A_2 ~B_2} = \eqt{A_1}{A_2}$. If $\nu(A_1)$, $\nu(A_2)$, $\nu(B_1)$, $\nu(B_2) \in \P \cup \E$, then there are four cases, corresponding to $\arrd$, $\impd$, $\blpi$ and $\fa$. For instance, if $\nu(A_1)$, $\nu(A_2)$, $\nu(B_1)$ and $\nu(B_2)$ are all generated by $\E$, then necessarily we have $\nu(A_1) = \El ~a_1$, $\nu(A_2) = \El ~a_2$, $\nu(B_1) = \El ~b_1$ and $\nu(B_2) = \El ~b_2$. Therefore $\eqt{\Pi x : A_1. ~B_1}{\Pi x : A_2. ~B_2} \coloneqq \heq{(a_1 \arrd (\lambda x : \El ~a_1. ~b_1))}{}{(a_2 \arrd (\lambda y : \El ~a_2. ~b_2))}{}$. The axiom is:
\[
  \begin{array}{lcl}
	\prod_{\arrd} &:&
	\Pi a_1, a_2 : \Set. ~\Pi b_1 : (\El ~a_1 \ra \Set). ~\Pi b_2 : (\El ~a_2 \ra \Set). ~\heq{a_1}{}{a_2}{}\\
	&& \ra (\Pi x : \El ~a_1. ~\Pi y : \El ~a_2. ~\heq{x}{}{y}{} \ra \heq{b_1 ~x}{}{b_2 ~y}{}) \\
	&& \ra \heq{(a_1 \arrd b_1)}{}{(a_2 \arrd b_2)}{}
  \end{array}
\]
Note that this axiom is derivable from the previous axioms. We proceed similarly for the cases $\impd$, $\blpi$ and $\fa$.

We write $\Sigma_{eq}$ for the signature formed by the axiom schemas defining the heterogeneous equality. Reflexivity, symmetry, and transitivity are standard axioms of equality. We have also added axioms stating that a heterogeneous equality comparing two objects of the same type acts like Leibniz equality. Finally, we have an axiom for the application constructor and one axiom for the abstraction constructor---that is functional extensionality. Both axioms are used by Oury~\cite{oury}, who also assumes the uniqueness of identity proofs principle that entails the Leibniz principle we use.

\section{Replacing Rewrite Rules}
\label{part_trans}

When working in theories with prelude encoding, rewriting originates from the rewrite rules of $\Sigma_{pre}$ (which are generic rewrite rules), from the rewrite rules $\Sigma_\T$ (which are defined by the user) and from $\beta$-reduction. The goal of this work is to replace the user-defined rewrite rules $\Sigma_\T$ by equational axioms. In the rest of the paper, we write $\vdashr$ for a derivation inside the source theory---the theory with user-defined rewrite rules---and $\vdash$ for a derivation inside the target theory---the theory with axioms instead of user-defined rewrite rules.

We now have all the tools to replace rewrite rules by equational axioms. To do so, we build suitable transports, such that if $\Gamma \vdash t : A$ and $\Gamma \vdash p : \eqt{A}{B}$, then $\Gamma \vdash \transp ~p ~t : B$. The goal is to insert such transports into the terms instead of using conversion with the rules of $\Sigma_\T$. In the signature, each rewrite rule $\ell \lra r$ is replaced by the equational axiom $\heq{\overbar{\ell}}{}{\overbar{r}}{}$.

\subsection{Transports}

If we have $\Gamma \vdash t : A$ and $\Gamma \vdash p : \eqt{A}{B}$, we want to transport $t$ from $A$ to $B$, that is to build a term $\transp ~p ~t$ such that $\Gamma \vdash \transp ~p ~t : B$. A paramount result is that $t$ and $\transp ~p ~t$ are heterogeneously equal.

\begin{lemma}[Transport]
\label{lemma_transp}
Given $\Gamma \vdash t : A$ and $\Gamma \vdash p : \eqt{A}{B}$ with $A$ and $B$ small types, there exists $\transp ~p ~t$, called transport of $t$ along $p$, such that:
\begin{itemize}
\item $\Gamma \vdash \transp ~p ~t : B$,
\item there exists $\eqtrans$ such that $\Gamma \vdash \eqtrans ~p ~t : \heq{\transp ~p ~t}{B}{t}{A}$.
\end{itemize}
\end{lemma}

\begin{proof}
$A$ and $B$ are small types and we have an equality $\eqt{A}{B}$. If $A,B \in \S$ then $\nu(A) = \nu(B) = A = B$ and we take $\transp ~p ~t \coloneqq t$ and $\eqtrans ~p ~t \coloneqq \refl_A ~t$. Otherwise, by construction of $\kappa$, we know that $\nu(A), \nu(B) \in \P$, or $\nu(A), \nu(B) \in \E$, and that $\nu(A)$ and $\nu(B)$ have the same structure. Moreover, using $A \equiv_{\beta\Sigma_{pre}} \nu(A)$, we have $\Gamma \vdash t : \nu(A)$. We proceed by induction on the grammar $\P$ (we proceed similarly for the grammar $\E$).
\begin{itemize}
\item If $\nu(A) = \Prf ~a$ and $\nu(B) = \Prf ~b$, then we have $\Gamma \vdash p : \heq{a}{}{b}{}$. We take $\transp ~p ~t \coloneqq \leibp_{\El ~\o} ~a ~b ~p ~(\lambda w : \El ~\o. ~w) ~t$. We conclude using $\eqleibp_{\El ~\o}$.

\item If $\nu(A) = A' \ra S$ and $\nu(B) = B' \ra S$, with $A', B' \in \P$ and $S \in \S$, then we have $\eqt{A'}{B'} = \eqt{A}{B}$. From $\Gamma \vdash p : \eqt{A'}{B'}$ we can build some $p'$ such that $\Gamma \vdash p' : \eqt{B'}{A'}$ (using $\sym$). By weakening, we also have $p' : \eqt{B'}{A'}$ in the context $\Gamma, m_b : B'$. By induction, we have $\transp ~p' ~m_b : A'$ and $\eqtrans ~p' ~m_b : \heq{\transp ~p' ~m_b}{}{m_b}{}$ in the context $\Gamma, m_b : B'$. We take $\transp ~p ~t \coloneqq \lambda m_b : B'. ~t ~(\transp ~p' ~m_b)$. Using $\trans$ and $\app$ we obtain an equality $\heq{t ~(\transp ~p' ~m_b)}{}{t ~m_a}{}$ in the context $\Gamma, m_a : A', m_b : B', p_m : \heq{m_a}{}{m_b}{}$. Using $\fun$ and $\equiv_{\beta\Sigma_{pre}}$, we have $\heq{\lambda m_b : B'.\ t\ (\transp\ p'\ m_b)}{}{t}{}$ in the context $\Gamma$.

\item If $\nu(A) = \Pi z : S. ~A'$ and $\nu(B) = \Pi z : S. ~B'$ with $A', B' \in \P$ and $S \in \S$, then we have $\eqt{A}{B} = \Pi z : S. ~\eqt{A'}{B'}$. By weakening and application, we have $\Gamma, z : S \vdash p ~z : \eqt{A'}{B'}$. By induction we have $\transp ~(p ~z) ~(t ~z) : B'$ and $\eqtrans ~(p ~z) ~(t ~z) : \heq{\transp ~(p ~z) ~(t ~z)}{}{t ~z}{}$ in the context $\Gamma, z : S$. We take $\transp ~p ~t \coloneqq \lambda z : S. ~\transp ~(p ~z) ~(t ~z)$. We obtain $\heq{\lambda z : S. ~\transp ~(p ~z) ~(t ~z)}{}{t}{}$ using $\fun$ and $\equiv_{\beta\Sigma_{pre}}$. \qed
\end{itemize}
\end{proof}
The transport of $t$ from $A$ to $B$ depends on the small form of $A$ and $B$. In that respect, there exists a different transport for each possible family of small form, and such transport is indexed over an equality of a small type.

\subsection{Translation of Terms}
\label{posstrans}

To translate a theory with rewrite rules into a theory with equational axioms, we add
transports at the proper locations in the terms and types. If we have $\Gamma \vdashr t : A$ in the source theory, we want to find $\overbar{\Gamma}$, $\overbar{t}$ and $\overbar{A}$ that are translations of $\Gamma$, $t$ and $A$, and such that $\overbar{\Gamma} \vdash \overbar{t} : \overbar{A}$ in the target theory.

We add transports in a term by induction on a typing derivation---which is not unique---so we may have different translations for a same term. As such, we define a relation $\trad$ where $\overbar{t} \trad t$ states that $\overbar{t}$ is a translation of $t$. The relation is defined by induction on the terms of \lpm. Variables, constants, $\Type$ and $\Kind$ are translations of themselves. The translations of $\lambda$-abstractions $\lambda x : A. ~t$, dependent types $\Pi x : A. ~B$ and applications $t ~u$ rely on the translations of $t$, $u$, $A$ and $B$. The most important part of the definition is that the translation is stable by transports: if $\overbar{t}$ is a translation of $t$, then $\transp ~p ~\overbar{t}$ is also a translation of $t$, with $p$ typically an equality. This relation captures all possible translations, but some are not correct as they may not be well typed. For instance, $\lambda x : \overbar{A}. ~\overbar{t}$ is not a valid translation of $\lambda x : A. ~t$ when the variable $x$ used in $\overbar{t}$ does not expect type $\overbar{A}$ but another translation $\overbar{A}'$.

\begin{definition}
The translation relation $\trad$ is defined by:
\begin{mathpar}
\inferrule*{ }{x ~\trad ~x}

\inferrule*{ }{c ~\trad ~c}

\inferrule*{ }{\Type ~\trad ~\Type}

\inferrule*{ }{\Kind ~\trad ~\Kind}

\inferrule*{\overbar{A} ~\trad ~A \\ \overbar{t} ~\trad ~t}{(\lambda x : \overbar{A}. ~\overbar{t}) ~\trad ~(\lambda x : A. ~t)}

\inferrule*{\overbar{A} ~\trad ~A \\ \overbar{B} ~\trad ~B}{(\Pi x : \overbar{A}. ~\overbar{B}) ~\trad ~(\Pi x : A. ~B)}

\inferrule*{\overbar{t} ~\trad ~t \\ \overbar{u} ~\trad ~u}{(\overbar{t} ~\overbar{u}) ~\trad ~(t ~u)}

\inferrule*{\overbar{t} ~\trad ~t}{(\transp ~p ~\overbar{t}) ~\trad ~t}
\end{mathpar}
where $p$ is an arbitrary term.
\end{definition}
Due to the typing rules of \lpm, transports for objects that have $\Type$ in their type do not exist. Therefore, the only well-typed translations of $\Type$, $\Kind$, $\Set$, $\Prf$ and $\El$ are themselves, and the well-typed translations of $\Pi x : A. ~B$ are of the form $\Pi x : \overbar{A}. ~\overbar{B}$ with $\overbar{A} ~\trad ~A$ and $\overbar{B} ~\trad ~B$.  It follows that a well-typed translation of a small type is still a small type. In particular, if $A \in \S$ then for any $\overbar{A}$ we have $\overbar{A} \coloneqq A$; if $\nu(A) \in \P$ then $\nu(\overbar{A}) \in \P$; and if $\nu(A) \in \E$ then $\nu(\overbar{A}) \in \E$.

We extend the relation to contexts and signatures. For each rewrite rule $\ell \lra r$ of a signature, we have $\bm{x} : \bm{B} \vdashr \ell : A$ and $\bm{x} : \bm{B} \vdashr r : A$, for some $\bm{B}$ and $A$, and some $\bm{x}$ representing the free variables of $\ell$. The translation of the rewrite rule $\ell \lra r$ is given by the equational axiom $\mathsf{eq}_{\ell r} : \Pi \bm{x} : \overbar{\bm{B}}. ~\heq{\overbar{\ell}}{\overbar{A}}{\overbar{r}}{\overbar{A}}$. Since the type of a term is not unique in \lpm, we have made a choice of $\bm{B}$ and $A$, which is not a problem as we will see in the proof of \cref{thm_elim}.
\begin{definition}
$\trad$ is defined on contexts and signatures by:
\begin{mathpar}
\inferrule*{ }{\langle \rangle ~\trad ~\langle \rangle}

\inferrule*{\overbar{\Gamma} ~\trad ~\Gamma \\ \overbar{A} ~\trad ~A}{(\overbar{\Gamma}, x : \overbar{A}) ~\trad ~(\Gamma, x : A)}

\inferrule*{\overbar{\Sigma} ~\trad ~\Sigma \\ \overbar{A} ~\trad ~A}{(\overbar{\Sigma}, c : \overbar{A}) ~\trad ~(\Sigma, c : A)}

\inferrule*{\overbar{\Sigma} ~\trad ~\Sigma \\ \overbar{\ell} ~\trad ~\ell \\ \overbar{r} ~\trad ~r \\ \overbar{\bm{B}} ~\trad ~\bm{B} \\ \overbar{A} ~\trad ~A}{(\overbar{\Sigma}, \mathsf{eq}_{\ell r} : \Pi \bm{x} : \overbar{\bm{B}}. ~\heq{\overbar{\ell}}{\overbar{A}}{\overbar{r}}{\overbar{A}}) ~\trad ~(\Sigma, \ell \lra r)}
\end{mathpar}
\end{definition}

\begin{lemma}
If $\overbar{t} ~\trad ~t$ and $\overbar{u} ~\trad ~u$ then $\overbar{t}[x \mapsto \overbar{u}] ~\trad ~t[x \mapsto u]$.
\end{lemma}

\begin{proof}
By induction on the derivation of $\overbar{t} ~\trad ~t$. For the case with the transport, we can prove that $(\transp ~p ~t)[x \mapsto u] = \transp ~p[x \mapsto u] ~t[x \mapsto u]$.
\qed
\end{proof}

\begin{definition}[Relation $\sim$]
We say that $t_1 \sim t_2$ when there exists some $t$ such that $t_1 ~\trad ~t$ and $t_2 ~\trad ~t$.
\end{definition}

\begin{lemma}
$\sim$ is an equivalence relation.
\end{lemma}

\begin{proof}
$\sim$ is reflexive, symmetric and transitive. When proving transitivity we exploit the fact that whenever $t ~\trad ~u_1$ and $t ~\trad ~u_2$, we have $u_1 = u_2$. Reflexivity is proved by induction on the term.
\qed
\end{proof}
An important result we need to prove is that two well-typed translations $t_1$ and $t_2$ of the same term $t$ are heterogeneously equal. By construction, both terms do not necessarily have the same type or the same context. We will always consider $\Gamma_1 \vdash t_1 : A_1$ and $\Gamma_2 \vdash t_2 : A_2$, where $\Gamma_1$ and $\Gamma_2$ have the same length and the same variables (with possibly different types). The equality between $t_1$ and $t_2$ must be typed in some context, but $\Gamma_1$ and $\Gamma_2$ are not sufficient. That is why we define a common context $\pack{\Gamma_1}{\Gamma_2}$ (written $\mathsf{Pack} ~\Gamma_1 ~\Gamma_2$ in the work of Winterhalter \al~\cite{transport}) by duplicating each variable and by assuming a witness of heterogeneous equality between these two duplicates. More precisely, we partially define $\pack{}{}$ by induction on small contexts:
\[ \pack{\langle \rangle}{\langle \rangle} \coloneqq \langle \rangle \]
\[ \pack{(\Gamma_1, x : A_1)}{(\Gamma_2, x : A_2)} \coloneqq \pack{\Gamma_1}{\Gamma_2}, x_1 : A_1[\gamma_1], x_2 : A_2[\gamma_2], p_x : \heq{x_1}{}{x_2}{} \]
where $\gamma_1$ substitutes variables $z$ by $z_1$ and $\gamma_2$ substitutes variables $z$ by $z_2$. We write $\gamma_{12}$ for the substitution that replaces the variables $z_1$ and $z_2$ by $z$ and the variable $p_z$ by $\refl ~z$.

\begin{lemma}
\label{lemma_2gamma}
If $\pack{\Gamma}{\Gamma} \vdash t : A$, then we can derive $\Gamma \vdash t[\gamma_{12}] : A[\gamma_{12}]$.
\end{lemma}

\begin{proof}
We proceed by induction on the length of $\Gamma$. If we have $\pack{\langle \rangle}{\langle \rangle} \vdash t : A$ then by definition we have $\langle \rangle \vdash t : A$. Suppose that we have $\pack{(\Gamma, x : B)}{(\Gamma, x : B)} \vdash t : A$. We apply successively \cref{lemma_subst} to replace $x_2$ and $x_1$ by $x$ and then $p_x$ by $\refl ~x$. \qed
\end{proof}
The following lemma states that two translations of a same term are heterogeneously equal.

\begin{lemma}[Equal translations]
\label{lemma_t1t2}
Let $t_1 \sim t_2$ such that $\Gamma_1 \vdash t_1 : A_1$ and $\Gamma_2 \vdash t_2 : A_2$ with $\Gamma_1$ and $\Gamma_2$ small contexts.
\begin{enumerate}
\item If $\Gamma_1 \vdash A_1 : \Type$ and $\Gamma_2 \vdash A_2 : \Type$, then there exists some $p$ such that $\pack{\Gamma_1}{\Gamma_2} \vdash p : \heq{t_1[\gamma_1]}{A_1[\gamma_1]}{t_2[\gamma_2]}{A_2[\gamma_2]}$.
\item If $t_1$ and $t_2$ are small types, then there exists some $p$ such that $\pack{\Gamma_1}{\Gamma_2} \vdash p : \eqt{t_1[\gamma_1]}{t_2[\gamma_2]}$.
\end{enumerate}
\end{lemma}

\begin{proof}

We proceed by induction on the derivation of $t_1 \sim t_2$. We show two interesting cases. The complete proof is available in \cref{appendix_t1t2}.
\begin{itemize}
\item \textsc{Transport} $(\transp ~p ~t_1) \sim t_2$

We have $\Gamma_1 \vdash \transp ~p ~t_1 : A_1$ and $\Gamma_2 \vdash t_2 : A_2$. By inversion of typing, we have $\Gamma_1 \vdash t_1 : A_1'$ and $\Gamma_1 \vdash p : \eqt{A_1'}{A_1}$. By induction there exists some $p_t$ such that $\pack{\Gamma_1}{\Gamma_2} \vdash p_t : \heq{t_1[\gamma_1]}{}{t_2[\gamma_2]}{}$. We also have $\Gamma_1 \vdash \eqtrans ~p ~t_1 : \heq{\transp ~p ~t_1}{}{t_1}{}$. We derive that $\pack{\Gamma_1}{\Gamma_2} \vdash (\eqtrans ~p ~t_1)[\gamma_1] : \heq{(\transp ~p ~t_1)[\gamma_1]}{}{t_1[\gamma_1]}{}$. We conclude using transitivity. \\

\item \textsc{Application} $(t_1 ~u_1) \sim (t_2 ~u_2)$

Suppose that $t_1 ~u_1$ and $t_2 ~u_2$ are small types. Then the only possible cases are $t_1 = t_2 = \Prf$ or $t_1 = t_2 = \El$. If $t_1 = t_2 = \Prf$, then we have $\Gamma_1 \vdash \Prf ~u_1 : \Type$ and $\Gamma_2 \vdash \Prf ~u_2 : \Type$. Since $\eqt{\Prf ~u_1}{\Prf ~u_2} = \heq{u_1}{}{u_2}{}$, the result is simply the induction hypothesis $\pack{\Gamma_1}{\Gamma_2} \vdash p : \heq{u_1[\gamma_1]}{}{u_2[\gamma_2]}{}$. We proceed similarly for $\El ~u_1 \sim \El ~u_2$.

Suppose that we have $\Gamma_1 \vdash t_1 ~u_1 : T_1$ and
$\Gamma_2 \vdash t_2 ~u_2 : T_2$ with $\Gamma \vdash T_1 : \Type$ and $\Gamma \vdash T_2 : \Type$. Then by inversion of typing we have
$\Gamma_1 \vdash u_1 : B_1$ and
$\Gamma_2 \vdash u_2 : B_2$ and
$\Gamma_1 \vdash t_1 : \Pi x : A_1. ~B_1$ and
$\Gamma_2 \vdash t_2 : \Pi x : A_2. ~B_2$, with $T_1 \equiv_{\beta\Sigma_{pre}} B_1[x \mapsto u_1]$ and $T_2 \equiv_{\beta\Sigma_{pre}} B_2[x \mapsto u_2]$.
By induction hypotheses, we have $\pack{\Gamma_1}{\Gamma_2} \vdash p_t : \heq{t_1[\gamma_1]}{}{t_2[\gamma_2]}{}$ and $\pack{\Gamma_1}{\Gamma_2} \vdash p_u : \heq{u_1[\gamma_1]}{}{u_2[\gamma_2]}{}$. We conclude using $\app$. \qed
\end{itemize}
\end{proof}

\subsection{Translation of Judgments}

In \cref{posstrans} we have seen all the possible translations for \emph{terms}. However, the only translations that matter are the translations of \emph{judgments}: context formation judgments and typing judgments.

\begin{definition}
For any $\vdashr \Gamma$ we define a set $\llbracket \vdashr \Gamma \rrbracket$ of valid judgments such that $\judgcont{\overbar{\Gamma}}{\Gamma}$ if and only if $\overbar{\Gamma} ~\trad ~\Gamma$.
For any $\Gamma \vdashr t : A$ we define a set $\llbracket \Gamma \vdashr t : A \rrbracket$ of valid judgments such that $\judg{\overbar{\Gamma}}{\overbar{t}}{\overbar{A}}{\Gamma}{t}{A}$ if and only if $\judgcont{\overbar{\Gamma}}{\Gamma}$, $\overbar{t} ~\trad ~t$ and $\overbar{A} ~\trad ~A$.
\end{definition}
We are now able to prove that it is possible to switch between two translations of a small type.

\begin{lemma}[Switching translations]
\label{lemma_2trad}
Suppose that we have $A$ a small type, $\judg{\overbar{\Gamma}}{\overbar{t}}{\overbar{A}}{\Gamma}{t}{A}$ and $\judg{\overbar{\Gamma}}{\overbar{A}'}{\Type}{\Gamma}{A}{\Type}$ with $\overbar{\Gamma}$ a small context. Then there exists $\overbar{t}'$ such that $\judg{\overbar{\Gamma}}{\overbar{t}'}{\overbar{A}'}{\Gamma}{t}{A}$.
\end{lemma}

\begin{proof}
If $\nu(A) \in \S$, then $\overbar{A} \coloneqq A$ and $\overbar{A}' \coloneqq A$, and we take $\overbar{t}' \coloneqq \overbar{t}$.
If $\nu(A) \in \P$, then $\nu(\overbar{A}), \nu(\overbar{A}') \in \P$ (this is similar for $\E$). As $\overbar{A}$ and $\overbar{A}'$ are two translations of $A$, we have $\overbar{A} \sim \overbar{A}'$. From \cref{lemma_t1t2}, we have $\pack{\overbar{\Gamma}}{\overbar{\Gamma}} \vdash p : \eqt{\overbar{A}[\gamma_1]}{\overbar{A}'[\gamma_2]}$. Using \cref{lemma_2gamma} we obtain $\overbar{\Gamma} \vdash p[\gamma_{12}] : \eqt{\overbar{A}}{\overbar{A}'}$. Using \cref{lemma_transp}, there exists some $\transp ~p[\gamma_{12}] ~\overbar{t} ~\trad ~t$ (since $\overbar{t} ~\trad ~t$) such that $\overbar{\Gamma} \vdash \transp ~p[\gamma_{12}] ~\overbar{t} : \overbar{A}'$. \qed
\end{proof}

\subsection{Translation of Theories}

Now that we have translated terms and judgments, we want to translate theories, so that the translation of every provable judgment in the source theory is provable in the target theory. The target theory $\T^{ax} = \Sigma_{pre} \cup \Sigma_{eq} \cup \overbar{\Sigma}_{\T}$ is obtained by adding the axioms of equality to the signature, and by translating $\Sigma_{\T}$. To do so, we translate each typed constant and rewrite rule one by one. At the end, the rewrite rules of $\Sigma_{\T}$ have been replaced by equational axioms.

The paramount result of this paper is the following theorem. The first item concerns context formation. The second item is about the translation of typing judgments. The third item focuses on convertible contexts. The fourth and fifth items are about the conversion rules. It is worth noting that in the second item we use the universal quantifier on $\overbar{\Gamma}$ instead of using the existential quantifier. We have opted for the universal quantifier so we can obtain the induction hypotheses for a common context.

\begin{theorem}[Elimination of the rewrite rules]
\label{thm_elim}
Let a theory $\T = \Sigma$ in \lpm such that $\T$ is a theory with prelude encoding and such that all the derivations considered are small derivations. There exists a signature $\overbar{\Sigma}_{\T} ~\trad ~\Sigma_{\T}$ such that the theory $\T^{ax} = \Sigma_{pre} \cup \Sigma_{eq} \cup \overbar{\Sigma}_{\T}$ satisfies:
\begin{enumerate}
\item If $\vdashr \Gamma$, then there exists $\judgcont{\overbar{\Gamma}}{\Gamma}$.

\item If $\Gamma \vdashr t : A$, then for every $\judgcont{\overbar{\Gamma}}{\Gamma}$ there exist $\overbar{t}$ and $\overbar{A}$ such that $\judg{\overbar{\Gamma}}{\overbar{t}}{\overbar{A}}{\Gamma}{t}{A}$.

\item If $(\vdashr \Gamma_1) \equiv (\vdashr \Gamma_2)$, then for every $\judgcont{\overbar{\Gamma}_1}{\Gamma_1}$ and $\judgcont{\overbar{\Gamma}_2}{\Gamma_2}$, we have $\vdash \pack{\overbar{\Gamma}_1}{\overbar{\Gamma}_2}$.

\item If $(\Gamma_1 \vdashr u_1 : A_1) \equiv (\Gamma_2 \vdashr u_2 : A_2)$ with $\Gamma_1 \vdashr A_1 : \Type$ and $\Gamma_2 \vdashr A_2 : \Type$, then for every $\judgcont{\overbar{\Gamma}_1}{\Gamma_1}$ and $\judgcont{\overbar{\Gamma}_2}{\Gamma_2}$, we have $\judg{\overbar{\Gamma}_1}{\overbar{u}_1}{\overbar{A}_1}{\Gamma_1}{u_1}{A_1}$ and $\judg{\overbar{\Gamma}_2}{\overbar{u}_2}{\overbar{A}_2}{\Gamma_2}{u_2}{A_2}$ and there exists some $p$ such that $\pack{\overbar{\Gamma}_1}{\overbar{\Gamma}_2} \vdash p : \heq{\overbar{u}_1[\gamma_1]}{\overbar{A}_1[\gamma_1]}{\overbar{u}_2[\gamma_2]}{\overbar{A}_2[\gamma_2]}$.

\item If $(\Gamma_1 \vdashr u_1 : \Type) \equiv (\Gamma_2 \vdashr u_2 : \Type)$, then for every $\judgcont{\overbar{\Gamma}_1}{\Gamma_1}$ and $\judgcont{\overbar{\Gamma}_2}{\Gamma_2}$, we have $\judg{\overbar{\Gamma}_1}{\overbar{u}_1}{\Type}{\Gamma_1}{u_1}{\Type}$ and $\judg{\overbar{\Gamma}_2}{\overbar{u}_2}{\Type}{\Gamma_2}{u_2}{\Type}$ and there exists some $p$ such that $\pack{\overbar{\Gamma}_1}{\overbar{\Gamma}_2} \vdash p : \eqt{\overbar{u}_1[\gamma_1]}{\overbar{u}_2[\gamma_2]}$.
\end{enumerate}
\end{theorem}

\begin{proof}
The proof of the five items is done by induction on the typing derivations, assuming the existence of $\overbar{\Sigma}_{\T}$. We show three relevant cases. The complete proof can be found in \cref{appendix_elim}.

\begin{itemize}
\item \textsc{Prod:}
\begin{mathpar}
\inferrule*{\Gamma \vdashr A : \Type \\ \Gamma, x : A \vdashr B : s}{\Gamma \vdashr \Pi x : A. ~B : s}
\end{mathpar}

Take $\judgcont{\overbar{\Gamma}}{\Gamma}$. By induction hypothesis, we have $\judg{\overbar{\Gamma}}{\overbar{A}}{\Type}{\Gamma}{A}{\Type}$. We have $(\overbar{\Gamma}, x : \overbar{A}) ~\trad ~(\Gamma, x : A)$ and we know that the only translation of sort $s$ is itself, therefore by induction hypothesis we have $\judg{\overbar{\Gamma}, x : \overbar{A}}{\overbar{B}}{s}{\Gamma, x : A}{B}{s}$. We conclude that $\overbar{\Gamma} \vdash \Pi x : \overbar{A}. ~\overbar{B} : s$ using the \textsc{Prod} rule. \\

\item \textsc{Conv:}
\begin{mathpar}
\inferrule*{\Gamma \vdashr t : A \\ (\Gamma \vdashr A : s) \equiv (\Gamma \vdashr B : s)}{\Gamma \vdashr t : B}
\end{mathpar}

Take $\judgcont{\overbar{\Gamma}}{\Gamma}$. As we consider small derivations, either $A$ is a small type or $A$ and $B$ are the same type.

If $A$ is a small type, then by induction hypothesis we have $\pack{\overbar{\Gamma}}{\overbar{\Gamma}} \vdash p : \eqt{\overbar{A}[\gamma_1]}{\overbar{B}[\gamma_2]}$. By \cref{lemma_2gamma} we obtain $\overbar{\Gamma} \vdash p[\gamma_{12}] : \eqt{\overbar{A}}{\overbar{B}}$. By \cref{lemma_2trad} and induction hypothesis we have $\judg{\overbar{\Gamma}}{\overbar{t}}{\overbar{A}}{\Gamma}{t}{A}$. Thanks to \cref{lemma_transp}, there exists some $\overbar{t}'$ such that $\judg{\overbar{\Gamma}}{\overbar{t}'}{\overbar{B}}{\Gamma}{t}{B}$.

If $A$ and $B$ are the same type, then no conversion is needed and the result is simply given the induction hypothesis $\overbar{\Gamma} \vdash \overbar{t} : \overbar{A}$. \\

\item \textsc{ConvRefl:}
\begin{mathpar}
\inferrule*{\Gamma \vdashr u : A}{(\Gamma \vdashr u : A) \equiv (\Gamma \vdashr u : A)}
\end{mathpar}

Take $\judgcont{\overbar{\Gamma}}{\Gamma}$. By induction hypothesis, we have $\judg{\overbar{\Gamma}}{\overbar{u}}{\overbar{A}}{\Gamma}{u}{A}$.

If $\Gamma \vdashr A : \Type$, then we build $\pack{\overbar{\Gamma}}{\overbar{\Gamma}} \vdash p : \heq{\overbar{u}[\gamma_1]}{}{\overbar{u}[\gamma_2]}{}$ using all the congruence rules of $\approx$.

We proceed similarly for the case $A = \Type$.
\end{itemize}
The existence of $\overbar{\Sigma}_{\T}$ is proved by induction on the length of $\Sigma_{\T}$, using the previous five items and $\langle \rangle ~\trad ~\langle \rangle$. \qed
\end{proof}

\begin{corollary}[Preservation]
If $\vdashr t : A$ and $\judg{}{A}{s}{}{A}{s}$, then there exists $\overbar{t}$ such that $\vdash \overbar{t} : A$.
\end{corollary}

\begin{proof}
By \cref{thm_elim} we have $\judg{}{\overbar{t}'}{\overbar{A}'}{}{t}{A}$. Using \cref{lemma_2trad} with $\overbar{A} \coloneqq A$, we have some $\overbar{t}$ such that $\judg{}{\overbar{t}}{A}{}{t}{A}$. \qed
\end{proof}
We directly derive the two following conservativity and consistency results. We say that a theory $\T_2$ is conservative over a theory $\T_1$ when every formula in the common language of $\T_1$ and $\T_2$ that is provable in $\T_2$ is also provable in $\T_1$.

\begin{corollary}[Conservativity]
$\T$ is a conservative extension of $\T^{ax}$.
\end{corollary}

\begin{corollary}[Relative consistency]
If $\T^{ax}$ is consistent then $\T$ is also consistent.
\end{corollary}

\section{Conclusion}

\paragraph{Discussion.} In this paper, we showed that it is possible to replace user-defined rewrite rules by equational axioms, in the case of the \lpcm. This result works for theories with prelude encoding---which is satisfied by expressive theories such as predicate logic and set theory---and for small derivations---which is in practice the case. So as to replace rewrite rules by equational axioms, we have defined a heterogeneous equality with standard axioms---reflexivity, symmetry, transitivity, Leibniz principle---and congruences for each constructor. At the end, the theory with rewrite rules is a conservative extension of the theory with axioms.

\paragraph{Related work.} The similar problem of the translation from an extensional system to an intensional system has been investigated by Oury~\cite{oury}. He proposed a translation from the Extensional Calculus of Constructions to the Calculus of Inductive Constructions with additional axioms that define a heterogeneous equality. Winterhalter, Sozeau and Tabareau provided a translation from extensional type theory to intensional type theory~\cite{transport,WinterhalterFormalMetaType}. They took advantage of the presence of dependent pairs to encode a heterogeneous equality, unlike Oury who defined it with axioms.

In this paper, we have shown the existence of a translation from a theory with rewrite rules to a theory with equational axioms. Technical challenges appear as we are not in an extensional type system. In particular, Oury and Winterhalter \al had a homogeneous equality in their source theory and introduce a heterogeneous equality in the target theory. In this work, the source theory does not contain a homogeneous equality, and the target theory only contains a heterogeneous equality.

The major difference with previous works is that we are in a logical framework without an infinite hierarchy of sorts $s_i : s_{i+1}$ for $i \in \mathbb{N}$. In \lpm, we only have $\Type : \Kind$, which is the reason why we cannot define an equality between types. As such an equality is of paramount importance in the transports, we have considered a subclass of types---called small types---for which we can define an equality. However, it is worth noting that the sorts of \lpm allowed a simplification: by construction, there is no transports on types, so the translation of a dependent function type is directly a dependent function type.

%
%
\bibliographystyle{splncs04}
\bibliography{biblio}

\newpage

\appendix
\section{Proof of \cref{lemma_t1t2}}
\label{appendix_t1t2}

\begingroup
\def\thelemma{\ref*{lemma_t1t2}}
\begin{lemma}[Equal translations]
Let $t_1 \sim t_2$ such that $\Gamma_1 \vdash t_1 : A_1$ and $\Gamma_2 \vdash t_2 : A_2$ with $\Gamma_1$ and $\Gamma_2$ small contexts.
\begin{enumerate}
\item If $\Gamma_1 \vdash A_1 : \Type$ and $\Gamma_2 \vdash A_2 : \Type$, then there exists some $p$ such that $\pack{\Gamma_1}{\Gamma_2} \vdash p : \heq{t_1[\gamma_1]}{A_1[\gamma_1]}{t_2[\gamma_2]}{A_2[\gamma_2]}$.
\item If $t_1$ and $t_2$ are small types, then there exists some $p$ such that $\pack{\Gamma_1}{\Gamma_2} \vdash p : \eqt{t_1[\gamma_1]}{t_2[\gamma_2]}$.
\end{enumerate}
\end{lemma}
\addtocounter{lemma}{-1}
\endgroup

\begin{proof}

We proceed by induction on the derivation of $t_1 \sim t_2$.
\begin{itemize}

\item \textsc{Variable} $x \sim x$

$p$ is given by the variable $p_x$ that belongs to $\pack{\Gamma_1}{\Gamma_2}$. \\

\item \textsc{Constant} $c \sim c$

If $c : A$ and $A : \Type$, then we take $p \coloneqq \refl ~c$.
The only constant $c : \Type$ is $\Set \in \S$. Therefore, we have nothing to prove. \\

\item \textsc{Transport} $(\transp ~p ~t_1) \sim t_2$

We have $\Gamma_1 \vdash \transp ~p ~t_1 : A_1$ and $\Gamma_2 \vdash t_2 : A_2$. By inversion of typing, we have $\Gamma_1 \vdash t_1 : A_1'$ and $\Gamma_1 \vdash p : \eqt{A_1'}{A_1}$. By induction there exists some $p_t$ such that $\pack{\Gamma_1}{\Gamma_2} \vdash p_t : \heq{t_1[\gamma_1]}{}{t_2[\gamma_2]}{}$. We also have $\Gamma_1 \vdash \eqtrans ~p ~t_1 : \heq{\transp ~p ~t_1}{}{t_1}{}$. We derive that $\pack{\Gamma_1}{\Gamma_2} \vdash (\eqtrans ~p ~t_1)[\gamma_1] : \heq{(\transp ~p ~t_1)[\gamma_1]}{}{t_1[\gamma_1]}{}$. We conclude using transitivity. \\

\item \textsc{Application} $(t_1 ~u_1) \sim (t_2 ~u_2)$

Suppose that $t_1 ~u_1$ and $t_2 ~u_2$ are small types. Then the only possible cases are $t_1 = t_2 = \Prf$ or $t_1 = t_2 = \El$. If $t_1 = t_2 = \Prf$, then we have $\Gamma_1 \vdash \Prf ~u_1 : \Type$ and $\Gamma_2 \vdash \Prf ~u_2 : \Type$. Since $\eqt{\Prf ~u_1}{\Prf ~u_2} = \heq{u_1}{}{u_2}{}$, the result is simply the induction hypothesis $\pack{\Gamma_1}{\Gamma_2} \vdash p : \heq{u_1[\gamma_1]}{}{u_2[\gamma_2]}{}$. We proceed similarly for $\El ~u_1 \sim \El ~u_2$.

Suppose that we have $\Gamma_1 \vdash t_1 ~u_1 : T_1$ and
$\Gamma_2 \vdash t_2 ~u_2 : T_2$ with $\Gamma \vdash T_1 : \Type$ and $\Gamma \vdash T_2 : \Type$. Then by inversion of typing we have
$\Gamma_1 \vdash u_1 : B_1$ and
$\Gamma_2 \vdash u_2 : B_2$ and
$\Gamma_1 \vdash t_1 : \Pi x : A_1. ~B_1$ and
$\Gamma_2 \vdash t_2 : \Pi x : A_2. ~B_2$, with $T_1 \equiv_{\beta\Sigma_{pre}} B_1[x \mapsto u_1]$ and $T_2 \equiv_{\beta\Sigma_{pre}} B_2[x \mapsto u_2]$.
By induction hypotheses, we have $\pack{\Gamma_1}{\Gamma_2} \vdash p_t : \heq{t_1[\gamma_1]}{}{t_2[\gamma_2]}{}$ and $\pack{\Gamma_1}{\Gamma_2} \vdash p_u : \heq{u_1[\gamma_1]}{}{u_2[\gamma_2]}{}$. We conclude using $\app$. \\

\item \textsc{Abstraction} $(\lambda x : A_1. ~t_1) \sim (\lambda x : A_2. ~t_2)$

Suppose that we have $\Gamma_1 \vdash \lambda x : A_1. ~t_1 : T_1$ and
$\Gamma_2 \vdash \lambda x : A_2. ~t_2 : T_2$. Then by inversion of typing we have $\Gamma_1 \vdash A_1 : \Type$ and
$\Gamma_2 \vdash A_2 : \Type$ and
$\Gamma_1, x : A_1 \vdash t_1 : B_1$ and
$\Gamma_2, x : A_2 \vdash t_2 : B_2$, with $T_1 \equiv_{\beta\Sigma_{pre}} \Pi x : A_1. ~B_1$ and $T_2 \equiv_{\beta\Sigma_{pre}} \Pi x : A_2. ~B_2$.
By induction hypothesis, we have $\pack{\Gamma_1}{\Gamma_2} \vdash p_A : \eqt{A_1[\gamma_1]}{A_2[\gamma_2]}$. By induction hypothesis, we have $p_t : \heq{t_1[\gamma_1, x \mapsto x_1]}{B_1[\gamma_1, x \mapsto x_1]}{t_2[\gamma_2, x \mapsto x_2]}{B_2[\gamma_2, x \mapsto x_2]}$ in the context $\pack{\Gamma_1}{\Gamma_2}, x_1 : A_1[\gamma_1], x_2 : A_2[\gamma_2], p_x : \heq{x_1}{}{x_2}{}$. We conclude using $\fun$. \\

\item \textsc{Product} $(\Pi x : A_1. ~B_1) \sim (\Pi x : A_2. ~B_2)$

If $\nu(\Pi x : A_1. ~B_1), \nu(\Pi x : A_2. ~B_2) \in \S$, then we take $\pack{\Gamma_1}{\Gamma_2} \vdash \lambda P : \El ~\o. ~\lambda h : \Prf ~P. ~h : \true$.

If $\nu(\Pi x : A_1. ~B_1), \nu(\Pi x : A_2. ~B_2) \in \P$ with $\nu(A_1)$, $\nu(A_2) \in \S$ and $\nu(B_1)$, $\nu(B_2) \in \P$, then the result can be derived from the induction hypothesis $\Gamma, x : \Set \vdash p_B : \eqt{B_1}{B_2}$ using $\fun$.

If $\nu(\Pi x : A_1. ~B_1), \nu(\Pi x : A_2. ~B_2) \in \P$ with $\nu(A_1)$, $\nu(A_2) \in \P$ and $\nu(B_1)$, $\nu(B_2) \in \S$, then the result is given by the induction hypothesis $\pack{\Gamma_1}{\Gamma_2} \vdash p_A : \eqt{A_1}{A_2}$.

If $\nu(\Pi x : A_1. ~B_1), \nu(\Pi x : A_2. ~B_2) \in \P$ with $\nu(A_1), \nu(A_2), \nu(B_1), \nu(B_2) \in \P$, then we necessarily have $\nu(A_1) = \Prf ~a_1$, $\nu(A_2) = \Prf ~a_2$, $\nu(B_1) = \Prf ~b_1$ and $\nu(B_2) = \Prf ~b_2$. We have
$\Gamma_1 \vdash A_1 : \Type$ and
$\Gamma_2 \vdash A_2 : \Type$ and
$\Gamma_1, x : A_1 \vdash B_1 : \Type$ and
$\Gamma_2, x : A_2 \vdash B_2 : \Type$ and
$\Gamma_1 \vdash \Pi x : A_1. ~B_1 : \Type$ and
$\Gamma_2 \vdash \Pi x : A_2. ~B_2 : \Type$.
By induction hypotheses, we have $\pack{\Gamma_1}{\Gamma_2} \vdash p_A : \heq{a_1[\gamma_1]}{}{a_2[\gamma_2]}{}$ and $\pack{\Gamma_1}{\Gamma_2}, x_1 : A_1[\gamma_1], x_2 : A_2[\gamma_2], p_x : \heq{x_1}{}{x_2}{} \vdash p_B : \heq{b_1[\gamma_1, x \mapsto x_1]}{}{b_2[\gamma_2, x \mapsto x_2]}{}$. Using the appropriate $\prod$, we are able to build equality $\eqt{\Pi x : A_1. ~B_1}{\Pi x : A_2. ~B_2} \coloneqq \heq{(a_1[\gamma_1] \impd (\lambda x_1 : \Prf ~a_1[\gamma_1]. ~b_1[\gamma_1, x \mapsto x_1]))}{}{(a_2[\gamma_2] \impd (\lambda x_2 : \Prf ~a_2[\gamma_2]. ~b_2[\gamma_2, x \mapsto x_2]))}{}$.

The other cases are treated similarly. \qed
\end{itemize}
\end{proof}

\section{Proof of \cref{thm_elim}}
\label{appendix_elim}

\begingroup
\def\thetheorem{\ref*{thm_elim}}
\begin{theorem}[Elimination of the rewrite rules]
Let a theory $\T = \Sigma$ in \lpm such that $\T$ is a theory with prelude encoding and such that all the derivations considered are small derivations. There exists a signature $\overbar{\Sigma}_{\T} ~\trad ~\Sigma_{\T}$ such that the theory $\T^{ax} = \Sigma_{pre} \cup \Sigma_{eq} \cup \overbar{\Sigma}_{\T}$ satisfies:
\begin{enumerate}
\item If $\vdashr \Gamma$, then there exists $\judgcont{\overbar{\Gamma}}{\Gamma}$.

\item If $\Gamma \vdashr t : A$, then for every $\judgcont{\overbar{\Gamma}}{\Gamma}$ there exist $\overbar{t}$ and $\overbar{A}$ such that $\judg{\overbar{\Gamma}}{\overbar{t}}{\overbar{A}}{\Gamma}{t}{A}$.

\item If $(\vdashr \Gamma_1) \equiv (\vdashr \Gamma_2)$, then for every $\judgcont{\overbar{\Gamma}_1}{\Gamma_1}$ and $\judgcont{\overbar{\Gamma}_2}{\Gamma_2}$, we have $\vdash \pack{\overbar{\Gamma}_1}{\overbar{\Gamma}_2}$.

\item If $(\Gamma_1 \vdashr u_1 : A_1) \equiv (\Gamma_2 \vdashr u_2 : A_2)$ with $\Gamma_1 \vdashr A_1 : \Type$ and $\Gamma_2 \vdashr A_2 : \Type$, then for every $\judgcont{\overbar{\Gamma}_1}{\Gamma_1}$ and $\judgcont{\overbar{\Gamma}_2}{\Gamma_2}$, we have $\judg{\overbar{\Gamma}_1}{\overbar{u}_1}{\overbar{A}_1}{\Gamma_1}{u_1}{A_1}$ and $\judg{\overbar{\Gamma}_2}{\overbar{u}_2}{\overbar{A}_2}{\Gamma_2}{u_2}{A_2}$ and there exists some $p$ such that $\pack{\overbar{\Gamma}_1}{\overbar{\Gamma}_2} \vdash p : \heq{\overbar{u}_1[\gamma_1]}{\overbar{A}_1[\gamma_1]}{\overbar{u}_2[\gamma_2]}{\overbar{A}_2[\gamma_2]}$.

\item If $(\Gamma_1 \vdashr u_1 : \Type) \equiv (\Gamma_2 \vdashr u_2 : \Type)$, then for every $\judgcont{\overbar{\Gamma}_1}{\Gamma_1}$ and $\judgcont{\overbar{\Gamma}_2}{\Gamma_2}$, we have $\judg{\overbar{\Gamma}_1}{\overbar{u}_1}{\Type}{\Gamma_1}{u_1}{\Type}$ and $\judg{\overbar{\Gamma}_2}{\overbar{u}_2}{\Type}{\Gamma_2}{u_2}{\Type}$ and there exists some $p$ such that $\pack{\overbar{\Gamma}_1}{\overbar{\Gamma}_2} \vdash p : \eqt{\overbar{u}_1[\gamma_1]}{\overbar{u}_2[\gamma_2]}$.
\end{enumerate}
\end{theorem}
\addtocounter{theorem}{-1}
\endgroup

\begin{proof}
From the source theory $\T = \Sigma_{pre} \cup \Sigma_{\T}$ we need to form the target theory $\T^{ax} = \Sigma_{pre} \cup \Sigma_{eq} \cup \overbar{\Sigma}_{\T}$, where $\overbar{\Sigma}_{\T} ~\trad ~\Sigma_{\T}$.

We prove the theorem in two steps: in a first step we admit the existence of $\overbar{\Sigma}_{\T}$ and we prove the five items; in a second step we prove the existence of $\overbar{\Sigma}_{\T}$.

Suppose that we have a translation $\overbar{\Sigma}_{\T}$. The five items of the theorem are proved by induction on the derivations inside $\T$.

We modify the derivations so that for each rewrite rule $\ell \lra r \in \Sigma_\T$, every instance of \textsc{ConvRule} is used with the same types $\vec{B}$ and $A$. This can be done using \textsc{Conv} and \textsc{ConvConv} rules. This pre-processing of the derivations simplify the induction case \textsc{ConvRule}.

\begin{itemize}

\item \textsc{Empty:}
\begin{mathpar}
\inferrule*{ }{\vdashr \langle \rangle}
\end{mathpar}

We have $\judgcont{\langle \rangle}{\langle \rangle}$. \\

\item \textsc{Decl:}
\begin{mathpar}
\inferrule*[right={$x \not\in \Gamma$}]{\vdashr \Gamma \\ \Gamma \vdashr A : s}{\vdashr \Gamma, x : A}
\end{mathpar}

By the first induction hypothesis, there exists $\judgcont{\overbar{\Gamma}}{\Gamma}$. Therefore, using the second induction hypothesis, there exists $\judg{\overbar{\Gamma}}{\overbar{A}}{s}{\Gamma}{A}{s}$. By definition, as $x \notin \Gamma$, we have $x \notin \overbar{\Gamma}$. Using the \textsc{Decl} rule, we can derive $\judgcont{\overbar{\Gamma}, x : \overbar{A}}{\Gamma, x : A}$. \\

\item \textsc{Sort:}
\begin{mathpar}
\inferrule*{\vdashr \Gamma}{\Gamma \vdashr \Type : \Kind}
\end{mathpar}

Suppose that there exists $\judgcont{\overbar{\Gamma}}{\Gamma}$. We conclude that $\overbar{\Gamma} \vdash \Type : \Kind$ using the \textsc{Sort} rule. \\

\item \textsc{Const:}
\begin{mathpar}
\inferrule*[right={$c : A \in \Sigma$}]{\vdashr \Gamma \\ \vdashr A : s}{\Gamma \vdashr c : A}
\end{mathpar}

Suppose that there exists $\judgcont{\overbar{\Gamma}}{\Gamma}$.

If $c : A \in \Sigma_{pre}$, then $\overbar{A} \coloneqq A$. Using the \textsc{Const} rule we derive that $\overbar{\Gamma} \vdash c : \overbar{A}$.

If $c : A \in \Sigma_{\T}$, then we have $c : \overbar{A} \in \overbar{\Sigma}_\T$ for some $\overbar{A} ~\trad ~A$. Using the \textsc{Const} rule we derive $\overbar{\Gamma} \vdash c : \overbar{A}$, since $\vdash \overbar{A} : s$. \\

\item \textsc{Var:}
\begin{mathpar}
\inferrule*[right={$x : A \in \Gamma$}]{\vdashr \Gamma}{\Gamma \vdashr x : A}
\end{mathpar}

Suppose that there exists $\judgcont{\overbar{\Gamma}}{\Gamma}$. As $x : A \in \Gamma$, we have $x : \overbar{A} \in \overbar{\Gamma}$ for some $\overbar{A} ~\trad ~A$. Then we derive that $\overbar{\Gamma} \vdash x : \overbar{A}$ using the \textsc{Var} rule.\\

\item \textsc{Prod:}
\begin{mathpar}
\inferrule*{\Gamma \vdashr A : \Type \\ \Gamma, x : A \vdashr B : s}{\Gamma \vdashr \Pi x : A. ~B : s}
\end{mathpar}

Suppose that there exists $\judgcont{\overbar{\Gamma}}{\Gamma}$. By induction hypothesis, we have $\judg{\overbar{\Gamma}}{\overbar{A}}{\Type}{\Gamma}{A}{\Type}$. We have $(\overbar{\Gamma}, x : \overbar{A}) ~\trad ~(\Gamma, x : A)$ and we know that the only translation of sort $s$ is itself, therefore by induction hypothesis we have $\judg{\overbar{\Gamma}, x : \overbar{A}}{\overbar{B}}{s}{\Gamma, x : A}{B}{s}$. We conclude that $\overbar{\Gamma} \vdash \Pi x : \overbar{A}. ~\overbar{B} : s$ using the \textsc{Prod} rule. \\

\item \textsc{Abs:}
\begin{mathpar}
\inferrule*{\Gamma \vdashr A : \Type \\ \Gamma, x : A \vdashr B : s \\ \Gamma, x : A \vdashr t : B}{\Gamma \vdashr \lambda x : A. ~t : \Pi x : A. ~B}
\end{mathpar}

Suppose that there exists $\judgcont{\overbar{\Gamma}}{\Gamma}$. By induction hypothesis, there exists $\judg{\overbar{\Gamma}}{\overbar{A}}{\Type}{\Gamma}{A}{\Type}$. As $(\overbar{\Gamma}, x : \overbar{A}) ~\trad ~(\Gamma, x : A)$ and by induction hypotheses, we have $\judg{\overbar{\Gamma}, x : \overbar{A}}{\overbar{B}}{s}{\Gamma, x : A}{B}{s}$ and $\judg{\overbar{\Gamma}, x : \overbar{A}}{\overbar{t}'}{\overbar{B}'}{\Gamma, x : A}{t}{B}$.

We know that the derivation is a small derivation. If $B$ is a small type, then we apply \cref{lemma_2trad} to obtain $\overbar{t}$ such that $\judg{\overbar{\Gamma}, x : \overbar{A}}{\overbar{t}}{\overbar{B}}{\Gamma, x : A}{t}{B}$. We conclude that $\overbar{\Gamma} \vdash \lambda x : \overbar{A}. ~\overbar{t} : \Pi x : \overbar{A}. ~\overbar{B}$ using the \textsc{Abs} rule.

If $B$ is not a small type, then the only other possibility is $t : B \in \Sigma_{pre}$. In that case we have $\overbar{B} = \overbar{B}' = B$. We conclude using the \textsc{Abs} rule. \\

\item \textsc{App:}
\begin{mathpar}
\inferrule*{\Gamma \vdashr t : \Pi x : A. ~B \\ \Gamma \vdashr u : A}{\Gamma \vdashr t ~u : B[x \mapsto u]}
\end{mathpar}

Suppose that there exists $\judgcont{\overbar{\Gamma}}{\Gamma}$. By the first induction hypothesis we have $\judg{\overbar{\Gamma}}{\overbar{t}}{\Pi x : \overbar{A}. ~\overbar{B}}{\Gamma}{t}{\Pi x : A. ~B}$. By the second induction hypothesis (and by \cref{lemma_2trad} if $A$ is a small type, we have $\judg{\overbar{\Gamma}}{\overbar{u}}{\overbar{A}}{\Gamma}{u}{A}$. We conclude that $\overbar{\Gamma} \vdash (\overbar{t} ~\overbar{u}) : \overbar{B}[x \mapsto \overbar{u}]$ with $\overbar{B}[x \mapsto \overbar{u}] ~\trad ~B[x \mapsto u]$ (since $\trad$ preserves substitution). \\

\item \textsc{Conv:}
\begin{mathpar}
\inferrule*{\Gamma \vdashr t : A \\ (\Gamma \vdashr A : s) \equiv (\Gamma \vdashr B : s)}{\Gamma \vdashr t : B}
\end{mathpar}

Suppose that there exists $\judgcont{\overbar{\Gamma}}{\Gamma}$.

As we consider small derivations, either $A$ is a small type or $A$ and $B$ are the same type.

If $A$ is a small type, then by induction hypothesis we have $\pack{\overbar{\Gamma}}{\overbar{\Gamma}} \vdash p : \eqt{\overbar{A}[\gamma_1]}{\overbar{B}[\gamma_2]}$. By \cref{lemma_2gamma} we obtain $\overbar{\Gamma} \vdash p[\gamma_{12}] : \eqt{\overbar{A}}{\overbar{B}}$. By \cref{lemma_2trad} and induction hypothesis we have $\judg{\overbar{\Gamma}}{\overbar{t}}{\overbar{A}}{\Gamma}{t}{A}$. Thanks to \cref{lemma_transp}, there exists some $\overbar{t}'$ such that $\judg{\overbar{\Gamma}}{\overbar{t}'}{\overbar{B}}{\Gamma}{t}{B}$.

If $A$ and $B$ are the same type, then no conversion is needed and the result is simply given the induction hypothesis $\overbar{\Gamma} \vdash \overbar{t} : \overbar{A}$. \\

\item \textsc{ConvRefl:}
\begin{mathpar}
\inferrule*{\Gamma \vdashr u : A}{(\Gamma \vdashr u : A) \equiv (\Gamma \vdashr u : A)}
\end{mathpar}

Suppose that there exists $\judgcont{\overbar{\Gamma}}{\Gamma}$. By induction hypothesis, we have $\judg{\overbar{\Gamma}}{\overbar{u}}{\overbar{A}}{\Gamma}{u}{A}$.

If $\Gamma \vdashr A : \Type$, then we build $\pack{\overbar{\Gamma}}{\overbar{\Gamma}} \vdash p : \heq{\overbar{u}[\gamma_1]}{}{\overbar{u}[\gamma_2]}{}$ using all the congruence rules of $\approx$.

We proceed similarly for the case $A = \Type$. \\

\item \textsc{ConvSymm:}
\begin{mathpar}
\inferrule*{(\Gamma \vdashr u : A) \equiv (\Gamma \vdashr v : B)}{(\Gamma \vdashr v : B) \equiv (\Gamma \vdashr u : A)}
\end{mathpar}

We use the induction hypothesis and the axiom $\sym$. \\

\item \textsc{ConvTrans:}
\begin{mathpar}
\inferrule*{(\Gamma \vdashr u : A) \equiv (\Gamma \vdashr v : B) \\ (\Gamma \vdashr v : B) \equiv (\Gamma \vdashr w : C)}{(\Gamma \vdashr u : A) \equiv (\Gamma \vdashr w : C)}
\end{mathpar}

Suppose that there exists $\judgcont{\overbar{\Gamma}}{\Gamma}$.

We are going to use an intermediary lemma, stating that if $(\Gamma \vdashr u : A) \equiv (\Gamma \vdashr v : B)$, then $\Gamma \vdashr A : \Type$ if and only if $\Gamma \vdashr B : \Type$, and $A = \Type$ if and only if $B = \Type$. This lemma can be proved by induction.

If $\Gamma \vdashr A : \Type$ and $\Gamma \vdashr C : \Type$, then $\Gamma \vdashr B : \Type$. By induction hypotheses, there exist $p_{AB}$ and $p_{BC}$ such that $\pack{\overbar{\Gamma}}{\overbar{\Gamma}} \vdash p_{AB} : \heq{\overbar{u}[\gamma_1]}{\overbar{A}[\gamma_1]}{\overbar{v}[\gamma_2]}{\overbar{B}[\gamma_2]}$ and $\pack{\overbar{\Gamma}}{\overbar{\Gamma}} \vdash p_{BC} : \heq{\overbar{v}'[\gamma_1]}{\overbar{B}'[\gamma_1]}{\overbar{w}[\gamma_2]}{\overbar{C}[\gamma_2]}$. Using \cref{lemma_t1t2}, there exists some $p$ such that $\pack{\overbar{\Gamma}}{\overbar{\Gamma}} \vdash p : \heq{\overbar{v}[\gamma_2]}{\overbar{B}[\gamma_2]}{\overbar{v}'[\gamma_1]}{\overbar{B}'[\gamma_1]}$. We conclude using transitivity.

The case $A = C = \Type$ is treated similarly. \\

\item \textsc{ConvDecl:}
\begin{mathpar}
\inferrule*[right={$x \notin \Gamma_1, \Gamma_2$}]{(\vdashr \Gamma_1) \equiv (\vdashr \Gamma_2) \\ (\Gamma_1 \vdashr A_1 : s) \equiv (\Gamma_2 \vdashr A_2 : s)}{(\vdashr \Gamma_1, x : A_1) \equiv (\vdashr \Gamma_2, x : A_2)}
\end{mathpar}

Suppose that there exist $\judgcont{\overbar{\Gamma}_1, x : \overbar{A}_1}{\Gamma_1, x : A_1}$ and $\judgcont{\overbar{\Gamma}_2, x : \overbar{A}_2}{\Gamma_2, x : A_2}$. Then we have $\judgcont{\overbar{\Gamma}_1}{\Gamma_1}$ and $\judgcont{\overbar{\Gamma}_2}{\Gamma_2}$, and by induction hypothesis $\pack{\overbar{\Gamma}_1}{\overbar{\Gamma}_2}$ is well formed. We conclude using \textsc{Decl} with $\pack{\overbar{\Gamma}_1}{\overbar{\Gamma}_2} \vdash x_1 : \overbar{A}_1[\gamma_1]$, with $\pack{\overbar{\Gamma}_1}{\overbar{\Gamma}_2} \vdash x_2 : \overbar{A}_2[\gamma_2]$, and with $\pack{\overbar{\Gamma}_1}{\overbar{\Gamma}_2} \vdash p_x : \heq{x_1}{\overbar{A}_1[\gamma_1]}{x_2}{\overbar{A}_2[\gamma_2]}$. \\

\item \textsc{ConvConst:}
\begin{mathpar}
\inferrule*[right={$c : A \in \Sigma$}]{(\vdashr \Gamma_1) \equiv (\vdashr \Gamma_2) \\ \vdashr A : s}{(\Gamma_1 \vdashr c : A) \equiv (\Gamma_2 \vdashr c : A)}
\end{mathpar}

Suppose that there exist $\judgcont{\overbar{\Gamma}_1}{\Gamma_1}$ and $\judgcont{\overbar{\Gamma}_2}{\Gamma_2}$. By induction hypothesis $\pack{\overbar{\Gamma}_1}{\overbar{\Gamma}_2}$ is well-formed.

If $c : A \in \Sigma_{pre}$, then $\overbar{A} \coloneqq A$, and if $c : A \in \Sigma_{\T}$, then we have $c : \overbar{A} \in \overbar{\Sigma}_\T$ for some $\overbar{A} ~\trad ~A$. Using the \textsc{Const} rule we derive that  $\overbar{\Gamma}_1 \vdash c : \overbar{A}$, that  $\overbar{\Gamma}_2 \vdash c : \overbar{A}$ and that $\pack{\overbar{\Gamma}_1}{\overbar{\Gamma}_2} \vdash c : \overbar{A}$. We obtain $\pack{\overbar{\Gamma}_1}{\overbar{\Gamma}_2} \vdash \refl_{\overbar{A}} ~c : \heq{c}{\overbar{A}}{c}{\overbar{A}}$. \\

\item \textsc{ConvVar:}
\begin{mathpar}
\inferrule*[right={$x : A_1 \in \Gamma_1, x : A_2 \in \Gamma_2$}]{(\vdashr \Gamma_1) \equiv (\vdashr \Gamma_2)}{(\Gamma_1 \vdashr x : A_1) \equiv (\Gamma_2 \vdashr x : A_2)}
\end{mathpar}

Suppose that there exist $\judgcont{\overbar{\Gamma}_1}{\Gamma_1}$ and $\judgcont{\overbar{\Gamma}_2}{\Gamma_2}$. By induction hypothesis $\pack{\overbar{\Gamma}_1}{\overbar{\Gamma}_2}$ is well-formed.

We have $x : A_1 \in \Gamma_1$ and $x : A_2 \in \Gamma_2$. Therefore we have $x_1 : \overbar{A}_1[\gamma_1] \in \pack{\overbar{\Gamma}_1}{\overbar{\Gamma}_2}$ and $x_2 : \overbar{A}_2[\gamma_2] \in \pack{\overbar{\Gamma}_1}{\overbar{\Gamma}_2}$. By construction we have the variable $p_x : \heq{x_1}{\overbar{A}_1[\gamma_1]}{x_2}{\overbar{A}_2[\gamma_2]}$ in the context $\pack{\overbar{\Gamma}_1}{\overbar{\Gamma}_2}$. We conclude using \textsc{Var}. \\

\item \textsc{ConvProd:}
\begin{mathpar}
\inferrule*{(\Gamma_1 \vdashr A_1 : \Type) \equiv (\Gamma_2 \vdashr A_2 : \Type) \\ (\Gamma_1, x : A_1 \vdashr B_1 : s) \equiv (\Gamma_2, x : A_2 \vdashr B_2 : s)}{(\Gamma_1 \vdashr \Pi x : A_1. ~B_1 : s) \equiv (\Gamma_2 \vdashr \Pi x : A_2. ~B_2 : s)}
\end{mathpar}

Suppose that there exist $\judgcont{\overbar{\Gamma}_1}{\Gamma_1}$ and $\judgcont{\overbar{\Gamma}_2}{\Gamma_2}$.

Since the derivation is a small derivation, either $\Pi x : A_1. ~B_1$ and $\Pi x : A_2. ~B_2$ are small types or are types of constants of $\Sigma_{pre}$.

If $\Pi x : A_1. ~B_1$ and $\Pi x : A_2. ~B_2$ are small types, then $A_1$, $B_1$, $A_2$ and $B_2$ are small types too. By induction hypotheses there exists $p_A$ such that $\pack{\overbar{\Gamma}_1}{\overbar{\Gamma}_2} \vdash p_A : \eqt{\overbar{A}_1[\gamma_1]}{\overbar{A}_2[\gamma_2]}$ and there exists $p_B$ such that $\pack{\overbar{\Gamma}_1}{\overbar{\Gamma}_2}, x_1 : \overbar{A}_1[\gamma_1], x_2 : \overbar{A}_2[\gamma_2], p_x : \heq{x_1}{}{x_2}{} \vdash p_B : \eqt{\overbar{B}_1[\gamma_1]}{\overbar{B}_2[\gamma_2]}$. If $\nu(A_1), \nu(A_2)$ and $\nu(B_1), \nu(B_2) \in \P$, we conclude using the appropriate $\prod$ axiom. The three other cases with grammars $\P$ and $\E$ are treated similarly.

If $\nu(A_1), \nu(A_2) \in \P$ or $\nu(A_1), \nu(A_2) \in \E$, but $\nu(B_1), \nu(B_2) \in \S$, or if $\nu(A_1), \nu(A_2) \in \S$ and $\nu(B_1), \nu(B_2) \in \P$ or $\nu(B_1), \nu(B_2) \in \E$, then we simply use the induction hypothesis. If $\nu(\Pi x : A_1. ~B_1),\nu(\Pi x : A_2. ~B_2) \in \S$, then we take $\pack{\overbar{\Gamma}_1}{\overbar{\Gamma}_2} \vdash \lambda P : \El ~\o. ~\lambda h : \Prf ~P. ~h : \true$.

If $\Pi x : A_1. ~B_1$ and $\Pi x : A_2. ~B_2$ are types of constants of $\Sigma_{pre}$, we necessarily have $\Pi x : A_1. ~B_1 = \Pi x : A_2. ~B_2$ and no conversion is needed. \\

\item \textsc{ConvAbs:}
\begin{mathpar}
\inferrule*{(\Gamma_1 \vdashr A_1 : \Type) \equiv (\Gamma_2 \vdashr A_2 : \Type) \\ (\Gamma_1, x : A_1 \vdashr B_1 : s) \equiv (\Gamma_2, x : A_2 \vdashr B_2 : s) \\ (\Gamma_1, x : A_1 \vdashr t_1 : B_1) \equiv (\Gamma_2, x : A_2 \vdashr t_2 : B_2)}{(\Gamma_1 \vdashr \lambda x : A_1. ~t_1 : \Pi x : A_1. ~B_1) \equiv (\Gamma_2 \vdashr \lambda x : A_2. ~t_2 : \Pi x : A_2. ~B_2)}
\end{mathpar}

Suppose that there exist $\judgcont{\overbar{\Gamma}_1}{\Gamma_1}$ and $\judgcont{\overbar{\Gamma}_2}{\Gamma_2}$.

Since the derivation is a small derivation, then $\Pi x : A_1. ~B_1$ and $\Pi x : A_2. ~B_2$ are necessarily small types. By induction hypotheses there exists $p_A$ such that $\pack{\overbar{\Gamma}_1}{\overbar{\Gamma}_2} \vdash p_A : \eqt{\overbar{A}_1[\gamma_1]}{\overbar{A}_2[\gamma_2]}$, and there exists $p_t$ such that $\pack{\overbar{\Gamma}_1}{\overbar{\Gamma}_2}, x_1 : \overbar{A}_1[\gamma_1], x_2 : \overbar{A}_2[\gamma_2], p_x : \heq{x_1}{}{x_2}{} \vdash p_t : \heq{\overbar{t}_1[\gamma_1, x \mapsto x_1]}{}{\overbar{t}_2[\gamma_2, x \mapsto x_2]}{}$. We conclude using the appropriate $\fun$. \\

\item \textsc{ConvApp:}
\begin{mathpar}
\inferrule*{(\Gamma_1 \vdashr t_1 : \Pi x : A_1. ~B_1) \equiv (\Gamma_2 \vdashr t_2 : \Pi x : A_2. ~B_2) \\ (\Gamma_1 \vdashr u_1 : A_1) \equiv (\Gamma_2 \vdashr u_2 : A_2)}{(\Gamma_1 \vdashr t_1 ~u_1 : B_1[x \mapsto u_1]) \equiv (\Gamma_2 \vdashr t_2 ~u_2 : B_2[x \mapsto u_2])}
\end{mathpar}

Suppose that there exist $\judgcont{\overbar{\Gamma}_1}{\Gamma_1}$ and $\judgcont{\overbar{\Gamma}_2}{\Gamma_2}$.

Since the derivation is a small derivation, either $\Pi x : A_1. ~B_1$ and $\Pi x : A_2. ~B_2$ are small types or are types of constants of $\Sigma_{pre}$.

Suppose that $\Pi x : A_1. ~B_1$ and $\Pi x : A_2. ~B_2$ are small types. By induction hypothesis there exists $p_t$ such that $\pack{\overbar{\Gamma}_1}{\overbar{\Gamma}_2} \vdash p_t : \heq{\overbar{t}_1[\gamma_1]}{}{\overbar{t}_2[\gamma_2]}{}$. By induction hypothesis and by \cref{lemma_2trad}, there exists $p_u$ such that $\pack{\overbar{\Gamma}_1}{\overbar{\Gamma}_2} \vdash p_u : \heq{\overbar{u}_1[\gamma_1]}{\overbar{A}_1[\gamma_1]}{\overbar{u}_2[\gamma_2]}{\overbar{A}_2[\gamma_2]}$. We conclude using $\app$.

Suppose that $t_1 = t_2 = \El$ or $t_1 = t_2 = \Prf$. The result is directly the induction hypothesis on $(\Gamma_1 \vdashr u_1 : A_1) \equiv (\Gamma_2 \vdashr u_2 : A_2)$.

Suppose that $t_1 = t_2 = \blpi$. We conclude using the induction hypothesis on $(\Gamma_1 \vdashr u_1 : A_1) \equiv (\Gamma_2 \vdashr u_2 : A_2)$, using $\refl ~\blpi$ and using $\app$. \\

\item \textsc{ConvBeta:}
\begin{mathpar}
\inferrule*{\Gamma \vdashr A : \Type \\ \Gamma, x : A \vdashr t : B \\ \Gamma, x : A \vdashr B : s \\ \Gamma \vdashr u : A}{(\Gamma \vdashr (\lambda x : A. ~t) ~u : B[x \mapsto u]) \equiv (\Gamma \vdashr t[x \mapsto u] : B[x \mapsto u])}
\end{mathpar}

Suppose that there exists $\judgcont{\overbar{\Gamma}}{\Gamma}$. By the induction hypotheses (and by \cref{lemma_2trad} if $A$ and $B$ are small types), we have $\judg{\overbar{\Gamma}}{\overbar{A}}{\Type}{\Gamma}{A}{\Type}$ and
$\judg{\overbar{\Gamma}, x : \overbar{A}}{\overbar{t}}{\overbar{B}}{\Gamma, x : A}{t}{B}$ and
$\judg{\overbar{\Gamma}, x : \overbar{A}}{\overbar{B}}{s}{\Gamma, x : A}{B}{s}$ and
$\judg{\overbar{\Gamma}}{\overbar{u}}{\overbar{A}}{\Gamma}{u}{A}$.
Using all the congruence rules of $\approx$ we have $\heq{((\lambda x : \overbar{A}. ~\overbar{t}) ~\overbar{u})[\gamma_1]}{}{((\lambda x : \overbar{A}. ~\overbar{t}) ~\overbar{u})[\gamma_2]}{}$. Using $\equiv_{\beta\Sigma_{pre}}$, we get $\heq{((\lambda x : \overbar{A}. ~\overbar{t}) ~\overbar{u})[\gamma_1]}{}{(\overbar{t}[x \mapsto \overbar{u}])[\gamma_2]}{}$. \\

\item \textsc{ConvRule:}
\begin{mathpar}
\inferrule*[right={$\ell \lra r \in \Sigma$}]{\vec{x} : \vec{B} \vdashr \ell : A \\ \vec{x} : \vec{B} \vdashr r : A \\ \Gamma \vdashr \vec{t} : \vec{B}}{(\Gamma \vdashr \ell[\vec{x} \mapsto \vec{t}] : A[\vec{x} \mapsto \vec{t}]) \equiv (\Gamma \vdashr r[\vec{x} \mapsto \vec{t}] : A[\vec{x} \mapsto \vec{t}])}
\end{mathpar}

Suppose that there exists $\judgcont{\overbar{\Gamma}}{\Gamma}$.

Suppose that we are in the case $\ell \lra r \in \Sigma_{pre}$. For example, consider the rewrite rule $\El ~(x \arrd y) \lra \Pi z : \El ~x. ~\El ~(y ~z)$ instantiated with $x \coloneqq u$ and $y \coloneqq v$. We have $x : \Set, y : \El ~x \ra \Set \vdash \El ~(x \arrd y) : \Type$ and $x : \Set, y : \El ~x \ra \Set \vdash \Pi z : \El ~x. ~\El ~(y ~z) : \Type$. By induction and \cref{lemma_2trad} we have $\judg{\overbar{\Gamma}}{\overbar{u}}{\Set}{\Gamma}{u}{Set}$ and $\judg{\overbar{\Gamma}}{\overbar{v}}{\El ~\overbar{u} \ra \Set}{\Gamma}{v}{\El ~u \ra Set}$. Since $\nu(\El ~(u \arrd v)) = \nu(\Pi z : \El ~u. ~\El ~(v ~z)) = u \arrd v$, we need to build $\heq{(\overbar{u} \arrd \overbar{v})[\gamma_1]}{}{(\overbar{u} \arrd \overbar{v})[\gamma_2]}{}$ in the context $\pack{\overbar{\Gamma}}{\overbar{\Gamma}}$. We build it using the congruence rules of $\approx$. The other cases are treated similarly.

Suppose that we are in the case $\ell \lra r \in \Sigma_\T$ with one free variable $x$ of type $B$. Then we have some $\mathsf{eq}_{\ell r} : \Pi x : \overbar{B}. ~\heq{\overbar{\ell}}{\overbar{A}}{\overbar{r}}{\overbar{A}}$ in the signature $\overbar{\Sigma}_\T$. By the third induction hypothesis and \cref{lemma_2trad}, we have $\judg{\overbar{\Gamma}}{\overbar{t}}{\overbar{B}}{\Gamma}{t}{B}$. By weakening and substitution, we have $\pack{\overbar{\Gamma}}{\overbar{\Gamma}} \vdash \overbar{t}[\gamma_1] : \overbar{B}$. We obtain $\pack{\overbar{\Gamma}}{\overbar{\Gamma}} \vdash \mathsf{eq}_{\ell r} ~\overbar{t}[\gamma_1] : \heq{\overbar{\ell}[x \mapsto \overbar{t}[\gamma_1]]}{}{\overbar{r}[x \mapsto \overbar{t}[\gamma_1]]}{}$. From the congruence rules of $\approx$ we build some $\heq{\overbar{r}[x \mapsto \overbar{t}[\gamma_1]]}{}{\overbar{r}[x \mapsto \overbar{t}[\gamma_2]]}{}$. We conclude by transitivity. The same proof can be done if we have multiple variables. \\

\item \textsc{ConvConv:}
\begin{mathpar}
\inferrule*{\Gamma \vdashr u  : A \\ (\Gamma \vdashr A : s) \equiv (\Gamma \vdashr B : s)}{(\Gamma \vdashr u  : A) \equiv (\Gamma \vdashr u : B)}
\end{mathpar}

Suppose that there exists $\judgcont{\overbar{\Gamma}}{\Gamma}$.

Since the derivation is a small derivation, either $A$ and $B$ are small types or are types of constants of $\Sigma_{pre}$.

If $A$ and $B$ are small types, then by the induction hypotheses, \cref{lemma_2trad} and \cref{lemma_2gamma} we have $\judg{\overbar{\Gamma}}{\overbar{u}}{\overbar{A}}{\Gamma}{u}{A}$ and $\overbar{\Gamma} \vdash p_{AB} : \eqt{\overbar{A}}{\overbar{B}}$. We conclude using \cref{lemma_transp}.

If $A$ and $B$ are types of constants of $\Sigma_{pre}$, we necessarily have $A = B$. Therefore we take $\pack{\overbar{\Gamma}}{\overbar{\Gamma}} \vdash \refl ~u : \heq{u}{A}{u}{A}$. The result works because the variables of $\overbar{\Gamma}$ do not appear in $u$ and $A$ so $u[\gamma_i] \coloneqq u$ and $A[\gamma_i] \coloneqq A$ for $i \in \{ 1, 2 \}$. \\
\end{itemize}

We just proved the five items assuming the existence of $\overbar{\Sigma}_{\T} ~\trad ~\Sigma_{\T}$. We now need to prove the existence of $\overbar{\Sigma}_{\T}$. We write $\Sigma^i_{\T}$ for the signature composed of the $i$-th first elements of $\Sigma_{\T}$. Formally, if $c : A$ is the $i+1$-th element of $\Sigma_{\T}$, then we have $\vdashr A : s$ in the theory $\Sigma_{pre} \cup \Sigma^i_{\T}$, and if $\ell \lra r$ is the $i+1$-th element of $\Sigma_{\T}$, then we have $\vec{x} : \vec{B} \vdashr \ell : A$ and $\vec{x} : \vec{B} \vdashr r : A$ in the theory $\Sigma_{pre} \cup \Sigma^i_{\T}$, for some $\vec{B}$ and $A$. We translate the signature $\Sigma_{\T}$ by induction:
\begin{itemize}
\item The translation of $\langle \rangle$ is simply $\langle \rangle$.
\item Suppose that we have the signature $\Sigma^i_{\T}, c : A$. By induction we get a translation $\overbar{\Sigma}^i_{\T} ~\trad ~\Sigma^i_{\T}$. We use the first step of the proof to get $\overbar{A} ~\trad ~A$ such that $\vdash \overbar{A} : s$. We have $(\Sigma^i_{\T}, c : \overbar{A}) ~\trad ~(\Sigma^i_{\T}, c : A)$. 
\item Suppose that we have the signature $\Sigma^i_{\T}, \ell \lra r$. By induction we get a translation $\overbar{\Sigma}^i_{\T} ~\trad ~\Sigma^i_{\T}$. We use the first step of the proof and \cref{lemma_2trad} to get $\vec{x} : \overbar{\vec{B}} \vdashr \overbar{\ell} : \overbar{A}$ and $\vec{x} : \overbar{\vec{B}} \vdashr \overbar{r} : \overbar{A}$. We have $(\overbar{\Sigma}^i_{\T}, \mathsf{eq}_{\ell r} : \Pi x : \overbar{\vec{B}}. ~\heq{\overbar{\ell}}{\overbar{A}}{\overbar{r}}{\overbar{A}}) ~\trad ~(\Sigma^i_{\T}, \ell \lra r)$.
\end{itemize}
At the end, we obtain a signature $\overbar{\Sigma}_{\T}$ that is a translation of $\Sigma_{\T}$, that does not contain rewrite rules, and such that $\T^{ax} = \Sigma_{pre} \cup \Sigma_{eq} \cup \overbar{\Sigma}_{\T}$ is a theory.
\qed
\end{proof}

\end{document}